\documentclass[runningheads]{llncs}

\usepackage[T1]{fontenc}
\usepackage{graphicx}

\let\llncssubparagraph\subparagraph
\let\subparagraph\paragraph
\usepackage[compact]{titlesec}
\let\subparagraph\llncssubparagraph
\titlespacing\section{0pt}{12pt plus 4pt minus 2pt}{4pt plus 2pt minus 2pt}
\titlespacing\subsection{0pt}{12pt plus 4pt minus 2pt}{4pt plus 2pt minus 2pt}
\titlespacing\subsubsection{0pt}{12pt plus 4pt minus 2pt}{4pt plus 2pt minus 2pt}
\expandafter\def\expandafter\normalsize\expandafter{%
    \normalsize%
    \setlength\abovedisplayskip{1ex}%
    \setlength\belowdisplayskip{1ex}%
    \setlength\abovedisplayshortskip{0.5ex}%
    \setlength\belowdisplayshortskip{0.5ex}%
}
\usepackage[nolist]{acronym}
\usepackage{amsfonts, amsmath, amssymb}
\usepackage[
boxed,
shortend, 
linesnumbered,
vlined
]{algorithm2e}
\usepackage{anyfontsize}
\usepackage{bm}
\usepackage{booktabs}
\usepackage[format=plain]{caption}
\usepackage{colortbl}
\usepackage{csquotes}
\usepackage{diagbox}
\usepackage{etoolbox}
\usepackage{hyphenat} 
\usepackage[shortlabels]{enumitem}
\usepackage{float}
\usepackage{mathtools}
\usepackage[defaultlines=2, all]{nowidow}
\usepackage{sectsty}
\usepackage[nodisplayskipstretch]{setspace} 
\usepackage{scrhack}
\usepackage{subcaption}
\usepackage{tabularx}
\usepackage{tikz}
\usepackage{xcolor}
\usepackage{todonotes}
\usepackage{pgfplots}
\usepackage[section]{placeins} 
\usepackage{xurl} 
\usetikzlibrary{arrows, shapes, calc, positioning}

\usepackage[hidelinks]{hyperref}
\usepackage[capitalize, noabbrev]{cleveref}

\graphicspath{{graphics/}}

\newcommand{\tikzdots}{\hspace*{1.75pt}$\dots$}

\definecolor{GreyClass2}{RGB}{189,189,189}
\definecolor{graphcolorpeach}{RGB}{254,232,200}
\definecolor{graphcolorpeach-orange}{RGB}{253,187,132}
\definecolor{graphcolordark-orange}{RGB}{227,74,51}
\definecolor{graphcolorred}{RGB}{215,25,28}
\definecolor{graphcoloralternative-peach-orange}{RGB}{253,174,97}
\definecolor{graphcoloralternative-peach-orange2}{RGB}{244,165,130}
\definecolor{battleshipgrey}{rgb}{0.52, 0.52, 0.51}
\definecolor{HKS44}{RGB}{0,84,159}
\definecolor{HKS44_75prcnt}{RGB}{64,127,183}
\definecolor{RWTH_Lightblue}{RGB}{142,189,229}
\definecolor{HKS44_25prcnt}{RGB}{199,221,242}
\definecolor{RWTH_Gold}{HTML}{F9A02C}
\definecolor{RWTH_Silver}{HTML}{797A7C}
\definecolor{RWTH_Magenta}{RGB}{227,0,102}
\definecolor{RWTH_Yellow}{RGB}{255,237,0}
\definecolor{RWTH_Petrol}{RGB}{0,97,101}
\definecolor{RWTH_Cyan}{RGB}{0,152,161}
\definecolor{RWTH_Green}{RGB}{87,171,39}
\colorlet{RWTH_Green_50prcnt}{RWTH_Green!50}
\definecolor{RWTH_Maygreen}{RGB}{189,205,0}
\colorlet{RWTH_Maygreen_50prcnt}{RWTH_Maygreen!50}
\definecolor{RWTH_Orange}{RGB}{246,168,0}
\definecolor{RWTH_Red}{RGB}{204,7,30}
\colorlet{RWTH_Red_50prcnt}{RWTH_Red!50}
\definecolor{RWTH_Bordeaux}{RGB}{161,16,53}
\definecolor{RWTH_Violett}{RGB}{97,33,88}
\definecolor{RWTH_Purple}{RGB}{122,111,172}
\definecolor{cutcolor1}{RGB}{240,249,232}
\definecolor{cutcolor2}{RGB}{186,228,188}
\definecolor{cutcolor3}{RGB}{123,204,196}
\definecolor{cutcolor4}{RGB}{67,162,202}
\definecolor{cutcolor5}{RGB}{8,104,172}
\definecolor{intervalcolor1} {RGB} {215,5,28}
\definecolor{intervalcolor2} {RGB} {253,174,97}
\definecolor{intervalcolor3} {RGB} {142,189,229} 
\definecolor{intervalcolor4} {RGB} {0,84,159} 
\definecolor{intervalcolor5} {RGB} {171,221,164}
\definecolor{arccolor1} {RGB} {142,189,229} 
\definecolor{arccolor2} {RGB} {166,86,40} 
\definecolor{arccolor3} {RGB} {247,129,191} 
\definecolor{arccolor4} {RGB} {228,26,28} 
\definecolor{arccolor5} {RGB} {0,84,159} 
\definecolor{arccolor6} {RGB} {152,78,163} 
\definecolor{arccolor7} {RGB} {0,66,37} 
\definecolor{arccolor8} {RGB} {255,127,0} 
\definecolor{classcolor1} {RGB} {44,127,184}
\definecolor{classcolor2} {RGB} {127,205,187}
\definecolor{classcolor3} {RGB} {237,248,177}
\definecolor{t1}{RGB}{166,206,227}
\definecolor{t2}{RGB}{178,223,138}
\definecolor{t3}{RGB}{31,120,180}
\definecolor{t4}{RGB}{51,160,44}
\definecolor{t5}{RGB}{253,191,111}
\definecolor{t6}{RGB}{255,127,0}
\definecolor{t7}{RGB}{202,178,214}
\definecolor{t8}{RGB}{177,89,40}
\definecolor{t9}{RGB}{227,26,28}
\definecolor{t10}{RGB}{255,237,111}
\definecolor{t11}{RGB}{251,154,153}
\definecolor{t12}{RGB}{106,61,154}
\newcommand{\eg}{e.\,g.,\ } 
\newcommand{\ie}{i.\,e.,\ } 
\newcommand{\wrt}{w.\,r.\,t.\ }
\newcommand{\Wlog}{W.\,l.\,o.\,g.,\ }
\newcommand{\WLOG}{w.\,l.\,o.\,g.\ }

\newcommand{\pushright}[1]{\ifmeasuring@#1\else\omit\hfill$\displaystyle#1$\fi\ignorespaces}
\newcommand{\pushleft}[1]{\ifmeasuring@#1\else\omit$\displaystyle#1$\hfill\fi\ignorespaces}

\def\hmath$#1${\texorpdfstring{{\rmfamily\textit{#1}}}{#1}}

\newcommand{\rom}[1]{\mathrm{\uppercase\expandafter{\romannumeral #1\relax}}}

\let\epsilon\varepsilon
\let\rho\varrho
\let\phi\varphi
\let\theta\vartheta

\renewcommand{\tilde}{\widetilde}

\DeclareMathOperator{\disjcup}{\mathbin{\dot{\cup}}} 
\newcommand{\generalGraph}{\mathcal{G}}

\newcommand{\pathGraph}{\mathcal{P}}

\newcommand{\AGR}[2]{
	\mathrm{\texttt{\upshape AGR#1[}}%
	#2%
	\mathrm{\texttt{\upshape]}}
}

\DeclareMathOperator{\pred}{pred}
\DeclareMathOperator{\psucc}{succ}

\DeclareMathOperator{\gap}{gap}

\DeclareMathOperator{\Option}{O}
\DeclareMathOperator{\Shortminus}{\! - \!}
\DeclareMathOperator{\Shortplus}{\! + \!}

\newcommand\mydots{\makebox[1em][c]{.\hfil.\hfil.\,}}

\newcommand{\Partition}[1]{
	\mathrm{Part(}#1\mathrm{)}
}

\newcommand{\revOpt}[2]{
	\mathcal{R}_{#1}(\Option_{#2})
}
\newcommand{\olb}{\alpha^*}
\newcommand{\oub}{\beta^*}
\newcommand{\price}{\pi}
\newcommand{\p}{p}

\AtEndEnvironment{proof}{\phantom{}\qed} 

\newcolumntype{Y}{>{\centering\arraybackslash}X}
\newcolumntype{L}[1]{>{\raggedright\arraybackslash}p{#1}} 
\newcolumntype{C}[1]{>{\centering\arraybackslash}p{#1}} 
\newcolumntype{R}[1]{>{\raggedleft\arraybackslash}p{#1}} 

\makeatletter
\patchcmd{\@algocf@start}
{-1.5em}
{-.75em}
{}{}
\makeatother

\newcommand{\algorithmFont}[1]{\small\ttfamily{#1}}

\SetAlCapSkip{1ex}
\SetAlFnt{\algorithmFont}
\SetNlSkip{0.5em}
\newcommand{\algorithmInit}{
	\setstretch{1.05}
	\DontPrintSemicolon%
	\SetNlSty{lineNumberFont}{}{}
	\SetKwProg{Fn}{function}{}{end}%
	\SetArgSty{argFont} 
	\SetKwFunction{KwFn}{function} 
	\SetKwSty{keywordFont}%
	\SetCommentSty{commentFont}%
	\SetKw{break}{break}%
	\SetInd{.4em}{.8em}
}
\SetKw{To}{to}
\SetKw{With}{with}
\SetKw{And}{and}
\SetKw{Or}{or}
\SetKw{Not}{not}
\SetKw{Downto}{downto}
\SetKw{DownTo}{downto}
\SetKw{Step}{step}
\SetKwRepeat{Do}{do}{while}

\let\oldnl\nl
\newcommand{\nonl}{\renewcommand{\nl}{\let\nl\oldnl}}

\spnewtheorem {observation}[lemma]{Observation}{\bfseries}{\normalfont}

\usepackage{thm-restate}

\setlist[description]{
	align = left,
	font={\bfseries\rmfamily},
	leftmargin=0pt,
	labelindent=0pt,
	listparindent=\parindent,
	labelwidth=0pt,
	itemindent=!,
	labelsep = 6pt,
	itemsep=2pt,
	topsep=6pt,
	parsep=2pt,
}

\makeatletter
\def\namedlabel#1#2{\begingroup
	#2%
	\def\@currentlabel{#2}%
	\phantomsection\label{#1}\endgroup
}
\makeatother

\input{TikzSet.tex}

\makeatletter
\patchcmd\@acf{\AC@acl}{\AC@foo}{}{}
\patchcmd\@acf{\AC@acl}{\AC@foo}{}{}
\patchcmd\@acf{\AC@foo}{\hskip\z@\AC@acl}{}{}
\patchcmd\@acf{\AC@foo}{\hskip\z@\AC@acl}{}{}
\makeatother

\begin{document}
	\SetKwFunction{StackVCP}{StackVCP}%
	\SetKwFunction{ComputeBounds}{ComputeBounds}%
	\SetKwFunction{CompareOptions}{CompareOptions}%
	\SetKwFunction{ComputeOptionsEOP}{ComputeOptionsEOP}%
	\SetKwFunction{ResolvePrices}{ResolvePrices}%
\begin{acronym}
	\acro{VC}{Vertex Cover}
	\acro{VC-Problem}[\textsc{VC}]{\textsc{Vertex Cover}}
	\acro{WVC}{Weighted Vertex Cover}
	\acro{MWVC}{minimum weight vertex cover}
	\acro{MST}[MST]{Minimum Spanning Tree}
	\acro{WVC-Problem}[\textsc{MWVC}]{\textsc{Minimum Weight Vertex Cover}}
	\acro{StackIS}[\textsc{StackIS}]{\textsc{Stackelberg Independent Set}}
	\acro{StackIntS}[\textsc{StackIvSch}]{\textsc{Stackelberg Interval Scheduling}}
	\acro{KP}[\textsc{KP}]{\textsc{Knapsack}}
	\acro{StackKP}[\textsc{StackKP}]{\textsc{Stackelberg Knapsack}}
	\acro{StackKP-W}[\textsc{StackKP-W}]{\textsc{Stackelberg Knapsack} (weight-controlled)}
	\acro{StackKP-P}[\textsc{StackKP-P}]{\textsc{Stackelberg Knapsack} (price-controlled)}
	\acro{MC}[\textsc{MC}]{\textsc{Max Closure}}
	\acro{StackMC}[\textsc{StackMC}]{\textsc{Stackelberg Max Closure}}
	\acro{StackMC-k}[$\textsc{StackMC}_k$]{\textsc{Stackelberg Max Closure}}
	\acro{MM}[\textsc{MaxM}]{\textsc{Maximum Matching}}
	\acro{StackMM}[\textsc{StackMaxM}]{\textsc{Stackelberg Maximum Matching}}
	\acro{MST-Problem}[\textsc{MST}]{\textsc{Minimum Spanning Tree}}
	\acro{StackMST}[\textsc{StackMST}]{\textsc{Stackelberg Minimum Spanning Tree}}
	\acro{SC}[\textsc{SC}]{\textsc{Set Cover}}
	\acro{StackSC}[\textsc{StackSC}]{\textsc{Stackelberg Set Cover}}
	\acro{StackSetPack}[\textsc{StackSetPack}]{\textsc{Stackelberg Set Packing}}
	\acro{SP}[\textsc{SP}]{\textsc{Shortest-Path}}
	\acro{StackSP}[\textsc{StackSP}]{\textsc{Stackelberg Shortest-Path}}
	\acro{StackSP-k}[$\textsc{StackSP}_k$]{\textsc{Stackelberg Shortest-Path}}
	\acro{SPT}[\textsc{SP}]{\textsc{Shortest-Path-Tree}}
	\acro{StackSPT}[\textsc{StackSPT}]{\textsc{Stackelberg Shortest\hyp{}Path\hyp{}Tree}}
	\acro{StackSPT-k}[$\textsc{StackSPT}_k$]{$\textsc{Stackelberg Shortest-Path-Tree}_k$}
	\acro{Stack}[\textsc{Stack}]{\textsc{Stackelberg Vertex Pricing Game}}
	\acro{Stack-k}[$\textsc{Stack}_k$]{\textsc{Stackelberg Vertex Pricing Game}}
	\acro{VC-Problem}[\textsc{VC}]{\textsc{Vertex Cover}}
	\acro{StackVC}[\textsc{StackVC}]{\textsc{Stackelberg Vertex Cover}}
	\acro{StackVCP}[\textsc{StackVCP}]{\textsc{Stackelberg Vertex Cover} on a path}
	\acro{StackVC-k}[$\textsc{StackVC}_k$]{\textsc{Stackelberg Vertex Cover}}
	\acro{3-SAT}{\textsc{3-Satisfiability}}
\end{acronym}
\acused{StackVC-k}
\acused{StackSP-k}
\acused{StackKP-W}
\acused{StackKP-P}


	\title{Stackelberg Vertex Cover on a Path\thanks{A shortened version of this paper appeared in the proceedings of the 16$^{\text{th}}$ Symposium on Algorithmic Game Theory (SAGT) \cite{StackVCP_SAGT23}.}}
	\author{Katharina Eickhoff\thanks{This work is funded by the Deutsche Forschungsgemeinschaft (DFG, German Research Foundation) – 2236/2} \and
		Lennart Kauther \and
		Britta Peis}
	\authorrunning{Eickhoff, Kauther, Peis}
	\institute{Chair of Management Science, RWTH Aachen University, Aachen, Germany\\
		\email{\{eickhoff,kauther,peis\}@oms.rwth-aachen.de}\\
		\url{https://www.oms.rwth-aachen.de/}
	}
	\maketitle
	\begin{abstract}
		A Stackelberg Vertex Cover game is played on an undirected graph $\generalGraph$ where some of the vertices are under the control of a \textit{leader}. The remaining vertices are assigned a fixed weight.
The game is played in two stages. First, the leader chooses prices for the vertices under her control. Afterward, the second player, called \textit{follower}, selects a min weight vertex cover in the resulting weighted graph. That is, the follower selects a subset of vertices $C^*$ such that every edge has at least one endpoint in $C^*$ of minimum weight w.r.t.\ to the fixed weights and the prices set by the leader. \acf{StackVC} describes the leader's optimization problem to select prices in the first stage of the game so as to maximize her revenue, which is the cumulative price of all her (priceable) vertices that are contained in the follower's solution. 
Previous research showed that \acs{StackVC} is \textsf{NP}-hard on bipartite graphs, but solvable in polynomial time in the special case of bipartite graphs, where all priceable vertices belong to the same side of the bipartition. In this paper, we investigate \ac{StackVC} on paths and present a dynamic program with linear time and space complexity.
		
	\end{abstract}

\section{Introduction}
	\label{sec:Introduction}
	Stackelberg games -- sometimes also called leader-follower games -- are a tremendously powerful framework to capture situations in which a dominant party called the \textit{leader} first makes a decision and afterward, the other players called \textit{followers} react to the leader's move. An important detail thereby is that the leader has  information about the followers' preferences and potential competitors. She thus tries to leverage this informational and temporal advantage to maximize her revenue. A typical example of such a scenario is a mono- or oligopolist dictating prices, and the consumers adapting their behavior accordingly. 
	The concept of Stackelberg games traces back to the German economist Heinrich Freiherr von Stackelberg \cite{Stackelberg1934} who also coined the term. 
	
	In many applications, the followers' preferences have a combinatorial structure and can be represented by a graph. 
	For instance, the edges of the graph may represent all possible means of transport in a city, some of which, \eg certain turnpikes, are under the control of the leader, and one follower wants to find a shortest or cheapest path between a designated source and destination. This would be an example of \ac{StackSP} with a single follower. The leader then obtains revenue corresponding to the price of every edge under her control that is included in the shortest path selected by the follower. Labbè et al. \cite{doi:10.1287/mnsc.44.12.1608} first introduced \ac{StackSP} to optimize turnpike returns. By doing so, they kicked off an entirely new branch of research on so-called \emph{Stackelberg network pricing games}. In these, the leader controls a subset of the vertices or edges on which she can set prices in the first stage of the game. The other resources are assigned fixed prices. 
	In the second stage of the game, one or several followers each select an optimal solution to the problem under consideration (e.g. \acl{SP}, \acl{MST-Problem}, \acl{WVC-Problem}, \ldots) based on the fixed prices, and the prices set by the leader. The goal of the leader is to select prices in the first stage of the game so as to maximize the resulting revenue, which is obtained from the prices of all resources she controls that are selected into the optimal solution(s) of the follower(s). As usual in the Stackelberg literature, we assume that the followers break ties in favor of the leader. Furthermore, we restrict our analysis to Stackelberg games in the single-follower setting.
	
	Other common combinatorial graph structures that have been investigated under the Stackelberg framework include \ac{MST-Problem} \cite{BILO2015643,10.1007/978-3-540-73951-7_7,Cardinal_StackelbergMST_PlanarTW}, \acl{MC} aka. project-selection \cite{StackMaxClosureOMS}, and \ac{WVC-Problem} \cite{BriestEtAl_NetworkPricing,StackMaxClosureOMS}.
	Note that in the latter two problems, the leader controls part of the vertices, while in \acs{StackSP} and \acs{StackMST} the leader controls part of the edges. We provide further details on related work below. Stackelberg network pricing games are notoriously hard, and so a striking commonality of prior research on Stackelberg network pricing games is that they mostly obtained negative complexity results. For example, \acs{StackSP} and  \textsc{StackMST} are known to be APX-hard, even in the single-follower setting \cite{10.1007/978-3-540-73951-7_7,Joret_2010}. The few positive complexity results are often either \textsf{FPT}-algorithms where the number of priceable entities is parameterized, \eg \cite{BILO2015643,10.5555/1071747.1071751}, or obtained under severe restrictions of the input.
	For instance, Briest et\ al.~\cite{BriestEtAl_NetworkPricing} showed that \ac{StackVC} is solvable on bipartite graphs if all priceable vertices lie on the same side of the bipartition. Their result was later improved by Baïou and Barahona~\cite{BipartiteVCpreflow}. In contrast, Jungnitsch et al.~\cite{StackMaxClosureOMS} showed that \ac{StackVC} -- the Stackelberg version of \ac{WVC-Problem} -- is \textsf{NP}-hard on bipartite graphs when the priceable vertices are allowed to lie on both sides of the bipartition. This left the question whether \ac{StackVC} is solvable in polynomial time on special classes of bipartite graphs, like paths and trees.
	
	\vfill
	\newpage
	
	In this paper, we show that \ac{StackVC} is solvable in linear time and storage on paths. It turns out that the algorithm and its analysis are quite involved; surprisingly so, given the extremely restrictive structure of a path. 
	Our main contribution is the following theorem:
	\begin{theorem}
		\label{thm:StackVCPlinTime}
		\acl{StackVC} can be solved on a path in strongly linear time and storage. 
	\end{theorem}
	\vspace*{-1.5em}
	
	\paragraph{Related work.}
	\label{subsec:relatedWork}
	In this paragraph, we provide a brief overview of previous research on combinatorial problems in the Stackelberg framework.
	
	Early research in the context of Stackelberg network pricing games was mostly concerned with edge-connectivity problems such as \ac{SP} \cite{10.1145/1386790.1386802,10.1007/978-3-642-17572-5_37,Joret_2010,doi:10.1287/mnsc.44.12.1608,10.5555/1071747.1071751}, \ac{MST-Problem} \cite{BILO2015643,10.1007/978-3-540-73951-7_7,Cardinal_StackelbergMST_PlanarTW}, and \acl{SPT} \cite{10.1007/978-3-540-92185-1_32,https://doi.org/10.48550/arxiv.1207.2317}. Most of the results are concerned with hardness even \wrt approximability. For instance, to our best knowledge, the strongest hardness result for \ac{StackSP} is the impossibility to approximate \ac{StackSP} by a factor better than (2-$\epsilon$) in polynomial time (assuming \textsf{P} $\neq$ \textsf{NP}) from Briest et al~\cite{10.1007/978-3-642-17572-5_37}. Similarly, Cardinal et al.~\cite{10.1007/978-3-540-73951-7_7} proved \textsf{APX}-hardness of \acs{StackMST}. There are also some results on the positive side but most of them require rather severe restrictions to the input. Possibly the most surprising general result is that a single-price strategy, \ie a pricing scheme that assigns a uniform price to every priceable entity, suffices to obtain a logarithmic approximation factor \cite{10.1145/1386790.1386802,BriestEtAl_NetworkPricing}. This approximation guarantee remains -- to our best knowledge -- best-to-date.
	
	A recent study by Cristi and Schroeder~\cite{CRISTI202299} analyzes the effect of allowing the leader to set negative prices in \ac{StackSP}. Since the influence of an edge goes beyond its direct neighbors, setting negative prices can actually increase the leader's overall revenue. Yet, it is easy to see that this is not true for \ac{StackVC}. Since including a vertex with weight less or equal than zero is always beneficial, a negative price cannot be better than selling the vertex for free. 
	
	Research on Stackelberg games is not limited to edge-selection problems, however. For instance, \cite{https://doi.org/10.1002/net.20457,PFERSCHY2019149,PFERSCHY202118} examine \ac{KP} while \cite{BipartiteVCpreflow,BriestEtAl_NetworkPricing,StackMaxClosureOMS} study \acf{WVC-Problem} in the Stackelberg setting. Since both these problems are already \textsf{NP}-hard in general \cite{Karp1972}, these publications either restrict to certain graph classes such as \ac{StackVC} on bipartite graphs, or in the case of \acs{StackKP}, assume that the follower uses an approximation algorithm to solve his problem.
	Böhnlein et al.~\cite{Boehnlein2021} extend the concept of Stackelberg games even further, allowing the follower to optimize an arbitrary convex program and thus capture non-binary preferences.
	
	As we stated above, we restrict our analysis to the single-follower setting. In fact, among the references mentioned above, only a few consider the multi-follower setting. Some examples of this are \cite{BILO2015643,BriestEtAl_NetworkPricing,CRISTI202299,StackMaxClosureOMS,doi:10.1287/mnsc.44.12.1608,10.5555/1071747.1071751}. Note that, in the multi-follower setting, one additionally needs to distinguish different variants of the model depending on whether or not the leader can sell copies of the same resource to multiple followers simultaneously and potentially obtain revenue for every copy. We refer to \cite{StackMaxClosureOMS} for further details on this. 
	\vspace*{-.75em}
	
\section{The Model}	
	\vspace*{-.5em}
	\label{sec:Model}
	Let $\generalGraph=(V,E)$ be an undirected graph. Recall that $C\subseteq V$ is a \emph{vertex cover} if every edge has at least one endpoint in $C$. Given weights $\omega \colon V\to \mathbb{R}_{\geq 0}$, a \emph{\ac{MWVC}} is a vertex cover $C^*$ of minimal weight $\omega(C^*):=\sum_{v\in C^*} \omega(v)$.
	
	In \acf{StackVC} the vertex set is partitioned into $V= F \disjcup P$, where $F$ is the set of \textit{fixed-price} vertices, and $P$ is the set of 
	\textit{priceable} vertices.
	Initially, weights $\omega \colon F\to \mathbb{R}_{\geq 0}$ are only assigned to the fixed-price vertices.
	The vertices in $P$ are under the control of the leader, who sets prices $\pi \colon P\to \mathbb{R}_{\geq 0}$ in the first stage of the game.
	Afterward, in the second stage of the game, the follower selects a \ac{MWVC} with respect to the weights $\omega \colon V\to \mathbb{R}_{\geq 0}$, where the weight of every priceable vertex corresponds to the price set by the leader, \ie  $\omega(p)=\pi(p)$ for all $p\in P$.		 		 
	
	In \ac{StackVC} the leader only obtains revenue for the priceable vertices contained in the follower's solution and aims to maximize her revenue. Both players are assumed to act fully rationally. Neighboring priceable vertices would force the follower to include at least one of them into a \ac{MWVC} and thus, allow the leader to generate an arbitrarily high revenue. We therefore restrict our analysis to instances without adjacent priceable vertices.
	Similarly, vertices with a negative price will be included in any \ac{MWVC} and can thus be removed by deleting all incident edges.
	
	To circumvent infinitely small $\epsilon$-values, we assume ties to be broken in favor of the leader, \ie the follower chooses a \ac{MWVC} that maximizes the leader's revenue. This assumption is common in Stackelberg literature. \acf{StackVC} then describes the optimization problem from the leader's perspective.  That is, for a game on graph $\generalGraph=(V = F \disjcup P,\,E,\,\omega)$, and initial weights $\omega: F \rightarrow \mathbb{R}$ on the fixed-price vertices, we want to solve:
	\[
	\max_{\pi\in \mathbb{R}^P} \left\{\sum_{p\in C^*\cap P} \pi(p) \; \Bigg \vert \; C^* \mbox{  is  MWVC of }\generalGraph \mbox{ with }\omega(p)=\pi(p)\ \text{ for all } p\in P\right\}.
	\]
	See Figure~\ref{fig:StackVC-Example} below for an exemplary instance of \ac{StackVC} on a path. 
	
	\paragraph{\ac{StackVC} on a path.}
	Unless stated otherwise, we always assume to be dealing with a path aka.\ line graph $\pathGraph =(V,E)$ with $n$ vertices $v_1,\dots,v_n$. \Wlog we choose an arbitrary orientation of $\pathGraph$ and number the vertices from left to right, \ie $\pathGraph = (v_1, \dots, v_n)$. This provides us with a total order $\preceq$ of the vertices in $V$. That is, we have $v_i \preceq v_j$ whenever $i \leq j$. In this case, we call $v_i$ a predecessor of $v_j$ and similarly $v_j$ a successor of $v_i$. In case we want to stress that $v_i \neq v_j$, \ie $i < k$, we write $v_i \prec v_j$.
	
	Let $\p_1,\dots, \p_k$ denote the set of priceable vertices. We assume the priceable vertices to be numbered by their appearance on $\pathGraph$, \ie $i < j$ if $p_i \prec p_j$. To improve readability, we may write $\price_i$ instead of $\price(\p_i)$ from hereon.
	
	We denote the direct predecessor of $v_i$ by $\pred(v_i)$, \ie $\pred(v_i) \coloneqq v_{i-1}$ for $i \geq 2$. 
	Similarly, we define the direct successor of $v_i$ by $\psucc(v_i) \coloneqq v_{i+1}$ for $i \leq n-1$.

	To simplify the notion of subpaths, we use $\pathGraph_{\prec v}$ to denote the subpath of $\pathGraph$ from $v_1$ to the direct predecessor $\pred(v)$ of $v = v_\ell $. That is,
	$\pathGraph_{\prec v} \coloneqq (v_1, \dots, v_{\ell-1})$. For the subpath containing all vertices until $v$, but including $v$, we use $\pathGraph_{\preceq v} \coloneqq (v_1, \dots, v_{\ell})$.
	Analogously, for a set $S \subseteq V$, we define the sets $S_{\preceq v}$ and $S_{v \preceq}$ [$S_{\prec v}$ and $S_{v \prec}$] as the elements in $S$ which are [proper] predecessors or successors of vertex $v$, respectively.
	Note that every path is a bipartite graph, where the parts correspond to the nodes with odd and even distance to $v_1$, respectively. We define $\Partition{v}$ as the set of vertices that are on the same side of the bipartition as $v$. Remark that we do not apply any restrictions on the distribution of $P$ \wrt the bipartition of $V$. 
	
	\begin{figure}[thbp]
	\centering
	\resizebox{.9\textwidth}{!}{
		\centering
		\begin{tikzpicture}[
			every node/.style ={
				minimum size=.55cm,
				inner sep = 1pt,
				font = \small,
				draw,
				fill = graphcolorpeach
			},
			label distance = -.1mm
			]
			\node [cover, aPartition, label = above:{\scriptsize $v_1$}] at (0,0) (v0) {$1$};
			\node [right = .53cm of v0, bPartition, label = above:{\scriptsize$v_2$}] (v1) {$5$};
			\node [cover, right = .53cm of v1, label = above:{\scriptsize$p_1$}, aPartition] (v2) {
				$\price_1$
			};
			\node [right = .53cm of v2, label = above:{\scriptsize$v_4$}, bPartition] (v3) {$9$};
			\node [cover, right = .53cm of v3, label = above:{\scriptsize$v_5$}, aPartition] (v4) {$8$};
			\node [right = .53cm of v4, label = above:{\scriptsize$p_2$}, bPartition] (v5) {
				$\price_2$
			};
			\node [cover, right = .53cm of v5, label = above:{\scriptsize$v_7$}, aPartition] (v6) {$3$};
			\node [cover, right = .53cm of v6, label = above:{\scriptsize$v_8$}, bPartition] (v7) {$2$};
			\node [aPartition, right = .53cm of v7, label = above:{\scriptsize$v_9$}] (v8) {$6$};
			
			\draw [] (v0) -- (v1);
			\draw [] (v1) -- (v2);
			\draw [] (v2) -- (v3);
			\draw [] (v3) -- (v4);
			\draw [] (v4) -- (v5);
			\draw [] (v5) -- (v6);
			\draw [] (v6) -- (v7);
			\draw [] (v7) -- (v8);
		\end{tikzpicture}%
	}
	\caption{A \ac{StackVC} instance with two priceable vertices. The \ac{MWVC} $C^*$ corresponding to the optimal pricing scheme $\mathbf{p}^* = (13, \infty)$ is highlighted in magenta.}
	\label{fig:StackVC-Example}
\end{figure}
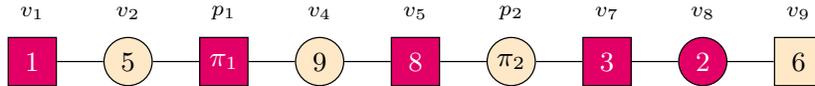
	
	\begin{example}
		\label{ex:StackVC-Example}
		We illustrate the intrinsic complexity of \ac{StackVC} on paths using the example depicted in Figure~\ref{fig:StackVC-Example}. 
		
		There are two priceable vertices $p_1$ and $p_2$ that lie on different sides of the partition which is indicated by the round and square vertices. Depending on the choice of $\price_1$ and $\price_2$ either $p_1$, $p_2$, both, or none of them are included in the follower's solution. Remark that despite its very limited size, finding a pricing scheme that yields the maximum revenue of 13, in this case, is far from trivial.
		
		For $\price_1 \leq 5$, there is a \ac{MWVC} including $\p_1$ independent of the choice of $\price_2$. This is because $\{v_1, \p_1, v_5\}$ is a \ac{MWVC} of $\pathGraph_{\preceq v_5}$ for this choice of $\price_1$. Conversely, for $\price_1 > 13$, $\{v_2, v_4\}$ is a \ac{MWVC} of $\pathGraph_{\preceq v_4}$ and thus no \ac{MWVC} of $\pathGraph$ contains $\p_1$. For any price in between, either $\p_1$ or $\p_2$ is included in the follower's solution. Therefore, it makes sense to evaluate the same cases for $\p_2$ for the choices $\price_1 = 5$ and $\price_1 = 13$. We refer to these prices as bounds and indeed they constitute the foundation of our algorithm for \ac{StackVC} on a path. Yet, one can clearly see that even with only two options per priceable vertex, simply evaluating every combination is not viable. By leveraging the problem's structure, we are, however, able to circumvent such an exponential blowup.
		
		For this particular instance, $\price_1 = 13,\, \price_2 \ge 11$ (e.g.\ $\price_2=\infty$) yields the maximum possible revenue of $13$. Note that under this pricing scheme, the leader does not obtain any revenue for $\p_2$.
	\end{example}
	
	\section{Sketch of the Algorithm}
	\label{sec:Algorithm}
	In this section, we sketch how \Cref{alg:StackVCP_sketch} (see pseudocode below) finds optimal prices for the leader given that the game is played on a path. In the subsequent sections, we describe the steps of the algorithm in further detail.
	
	Before we start describing our algorithm, let us shortly think how a \ac{MWVC} of a path can be computed in the setting where \emph{all} vertices are fixed-price vertices. Clearly, the problem is polytime solvable, since bipartite graphs allow for an efficient computation of a \ac{MWVC} (Theorem of K\H{o}nig-Egerv\'ary \cite{egervary1931kombinatorische,KoenigThm1931}), and clearly every path is bipartite.
	In fact, one can even compute a \ac{MWVC} of a path in linear time via the following dynamic programming approach:
	Traverse the path $\pathGraph$ from left to right and for every node $v$, check if there is a cheapest cover of the subpath $\pathGraph_{\preceq v}$ that includes $v$. If so, cut the path and repeat the procedure for the remaining path. The set $C^*$ containing $v$ and every second predecessor of it is a \ac{MWVC} of $\pathGraph_{\preceq v}$. Note that, if the respective cheapest covers of $\pathGraph_{\preceq v}$ in- and excluding $v$ have the same cost, we may either cut or continue.
	
	In the Stackelberg setting, however, traversing the path in one direction and setting the prices does not yield optimal prices in general.
	Indeed, the leader obtains revenue for vertex $\p$ if there is a \ac{MWVC} of the subpath up to $\p$ that includes $\p$. Still, in most cases, $\p$ is included in a cover of the entire path for a price beyond that. This can be nicely observed in the instance depicted in Figure~\ref{fig:StackVC-Example}.  The maximum price for $\p_1$ such that it is included in a \ac{MWVC} of $\pathGraph_{\preceq p_1}$ is $4$. Yet, it can be included for a price as high as $13$ (when choosing $\price_2$ accordingly).
	The primary challenge in \ac{StackVC} is, however, that the priceable vertices influence each other. If all priceable vertices are on one side of the bipartition, there is always an optimal solution including all priceable vertices. This vastly simplifies the problem. In fact, it turns the problem from \textsf{NP}-hard \cite{StackMaxClosureOMS} to polytime solvable \cite{BriestEtAl_NetworkPricing} on general bipartite graphs. 
	The reason for this is that for instances with priceable vertices on both sides of the bipartition, the leader must first decide which priceable vertices she wants to be included in the follower's solution and then maximize her revenue given this selection. Her problem thus becomes twofold: Find an optimal selection of priceable vertices and then maximize prices \wrt this selection. The instance discussed in Example~\ref{ex:StackVC-Example} illustrates this nicely.
	\vspace*{-1em}
	
	\paragraph{Overview.}
	Due to the challenges described above, simple dynamic programming is not sufficient for solving \ac{StackVC} on a path. 
	Nevertheless, we can employ it to compute lower and upper bounds for the priceable vertices (cf.\ Section~\ref{sec:cutting+bounds}). Note that the bounds of $\p_i$ depend on the prices $\price_1, \dots, \price_{i-1}$ chosen for the previous priceable vertices. Consequently, a careful analysis is necessary, as we discuss later. We use these bounds for a benchmark case to tackle the combinatorial challenge of finding the optimal selection of priceable vertices that the leader wants the follower to include in his \ac{MWVC}. We define the lower bound in a way such that matching the lower bound of a priceable vertex $\p_i$ guarantees that it is included in the follower's solution, \ie the leader obtains the price of $\p_i$ as revenue. When setting the price to match the upper bound, $\p_i$ can still be included by reducing the revenue obtained for the next priceable vertex $\p_{i+1}$.  However, this is not always desirable. Therefore, we distinguish these options for each priceable vertex (cf.\ Section~\ref{subsec:options}). To avoid an exponential blowup, we perform a sophisticated analysis to discard suboptimal options as soon as possible (cf.\ Sections ~\ref{subsec:AppraisingRevenue} and~\ref{subsec:comparing_options}). By doing so, we can restrict ourselves to two potential pricing schemes at any time, giving us the desired linear complexity. In the end, we calculate the optimal prices for the leader by choosing the prices according to the optimal pricing scheme (given by the optimal option for each priceable vertex).
	
	\begin{algorithm}[tbph]
	\algorithmInit
	\Fn{\StackVCP{$\pathGraph$}}{	
		\KwIn{A path $\pathGraph = (v_{1}, \dots, v_{n})$ with $k$ priceable vertices $p_1, \dots, p_k$ and vertex weights $\omega$ for the fixed-price vertices.}
		\KwOut{A revenue-maximizing pricing scheme $\price^* \in \mathbb{R}^k$.}
		\For{$i=1$ \To $k$}{
			$\olb[i]$, $\oub[i]$ = \ComputeBounds{$\pathGraph$, $\olb[1]$, $\dots$, $\olb[i-1]$} \hspace{-3ex} \tcp*[r]{cf.\ Section~\ref{sec:cutting+bounds}}
		}
		\tcp*[r]{compute optimal selection of priceable vertices}
		opt = \CompareOptions{$\pathGraph$, $\olb$, $\oub$}  \hspace{-10ex} \tcp*[r]{cf.\ Section~\ref{subsec:comparing_options}}
		\tcp*[r]{Compute optimal prices for given selection}
		$\price_1^*, \dots, \price_k^*$ = \ResolvePrices{$\pathGraph$, \textup{\texttt{opt}}} \tcp*[r]{cf.\ Section~\ref{subsec:comparing_options}}
		\KwRet{$\price^*$}
	}
	\setstretch{1}
	\caption{The outline of our algorithm for \ac{StackVC} on a path.}
	\label{alg:StackVCP_sketch}
\end{algorithm}
	
	\section{Structure of Minimum Weight Vertex Covers}
	\label{sec:in_cover}
	In this section, we delve into the general structure of \acp{MWVC} of a path. We start with a simple lemma characterizing the cornerstones of such a cover.
	\begin{lemma}[Cover condition]
		\label{lem:in_cover}
		Let $v$ be the first vertex on $\pathGraph$ for which the following condition holds.
		\begin{equation}
			\label{eq:cover_condition}
			\tag{cover condition}
			\qquad \qquad \sum_{u \in \pathGraph_{\preceq v} \cap \Partition{v}} \omega(u) \leq \sum_{u \in \pathGraph_{\preceq v} \setminus \Partition{v}} \omega(u).
		\end{equation}
		Then, $C^* \coloneqq \Partition{v} \cap \pathGraph_{\preceq v}$ is a \ac{MWVC} of $\pathGraph_{\preceq v}$. Moreover, there is a \ac{MWVC} of $\pathGraph$ containing $C^*$.
		If the inequality is strict, $C^*$ is the unique \ac{MWVC} of $\pathGraph_{\preceq v}$ and every \ac{MWVC} of $\pathGraph$ includes $C^*$.
	\end{lemma}
	
	\noindent
	We call the inequality \textit{(weak) \ref{eq:cover_condition}} because it allows us to construct a cover for $\pathGraph_{\preceq v}$. We speak of the \textit{strict \ref{eq:cover_condition}} whenever the inequality is strict and thus, $C^*$ is contained in any \ac{MWVC} of the entire path $\pathGraph$. In the Stackelberg setting, fulfilling the weak \ref{eq:cover_condition} often suffices to achieve the desired effect in terms of in- or excluding a specific vertex since we generally assume the follower to pick a \ac{MWVC} in favor of the leader.
	\begin{proof}
		Observe that $C^*$ and $\pathGraph_{\preceq v} \setminus C^*$ are the only vertex covers of  $\pathGraph_{\preceq v}$ that do not contain any neighboring vertices. Since the \ref{eq:cover_condition} holds at vertex $v$, we get
		\[
		\omega(C^*) \leq \omega(\pathGraph_{\preceq v}  \setminus C^*).
		\]
		
		For the sake of contradiction, assume there was a \ac{MWVC} $C$ of $\pathGraph_{\preceq v}$ containing two neighboring vertices. \Wlog let $x \prec y$ be the first neighboring vertices in $C$. Then,  $C' \coloneqq \bigl(\pathGraph_{\prec x} \cap \Partition{y}\bigr) \cup C_{y \preceq}$ is a vertex cover of $\pathGraph_{\preceq v}$ which is strictly cheaper than $C$, since
		\begin{align*}
			\omega\bigl(\pathGraph_{\prec x} \cap \Partition{y}\bigr)  < \omega\bigl(\pathGraph_{\prec x} \cap \Partition{x}\bigr)  = \omega \bigl(C_{\preceq x}\bigr),
		\end{align*}		
		where the inequality holds since $x$ does not satisfy the \ref{eq:cover_condition}.
		Hence, $C^*$ is a \ac{MWVC} of $\pathGraph_{\preceq v}$ which contains $v$. Furthermore, $C^*$ is the unique minimal cover if the strict \ref{eq:cover_condition} is fulfilled.
		
		Since a vertex cover of the entire path $\pathGraph$ must cover the edges in $\pathGraph_{\preceq v} \cup \{v, \psucc(v)\} \cup \pathGraph_{v \prec}$, each \ac{MWVC} $C_{v \prec}^*$ of $\pathGraph_{v \prec}$ completes $C^*$ to a \ac{MWVC} of $\pathGraph$ since $C^*$ already covers $\pathGraph_{\preceq v} \cup \{v, \psucc(v)\}$.
		
		To see that $C^*$ is contained in every \ac{MWVC} $\tilde{C}$ in case of strict inequality, note that $\tilde{C}$ has to cover $\pathGraph_{\preceq v}$ as well. Moreover, if $\tilde{C}_{v \preceq} \neq C^*$ replacing it with $C^*$ would result in a cheaper vertex cover.
	\end{proof}	
	
	\begin{corollary}
		\label{cor:cutting}
		We can split an instance $\pathGraph$ at any node $v$ satisfying the (weak) \ref{eq:cover_condition}. That is, we can cut the path after $v$, find a \ac{MWVC} for $\pathGraph_{\preceq v}$ including~$v$, and independently compute a \ac{MWVC} for $\pathGraph_{v \prec}$.
	\end{corollary}	
	Given that, the following statement follows immediately.
	\begin{corollary}
		\label{cor:VC}
		Let $v_1 \prec \dots \prec v_\ell$ be all the vertices on $\pathGraph_{\preceq v_{\ell}}$ satisfying the \ref{eq:cover_condition} such that only $v_\ell$ may satisfy the strict \ref{eq:cover_condition}. Let $S_i \coloneqq \{u \in \pathGraph \mid v_{i-1} \prec u \prec v_{i}\}$ denote the set of vertices between $v_{i-1}$ and $v_i$. Furthermore, let $C^*$ be a \ac{MWVC} of $\pathGraph$. Then, $(S_i \cap \Partition{v_i}) \subseteq C^*$ if and only if $v_i \in C^*$. Moreover, if $v_{\ell}$ satisfies the strict \ref{eq:cover_condition}, $v_{\ell}$ is contained in every \ac{MWVC} of $\pathGraph$.
	\end{corollary}
	
	In the proof of Lemma~\ref{lem:in_cover}, we use that $v$ is the first vertex satisfying the (weak) \ref{eq:cover_condition}. Yet, it suffices that no predecessor of $v$ fulfills the strict \ref{eq:cover_condition} to obtain a \ac{MWVC} of $\pathGraph$ containing $v$. In that case, the cover may not be unique though.
	
	Our algorithm cuts the instance whenever there is a vertex $v$ satisfying the strict \ref{eq:cover_condition}. It then treats the resulting instances $\pathGraph_{\preceq v}$ and $\pathGraph_{v \prec}$ separately.
	
	To get a better understanding of the above statements, it may be helpful to consider the path depicted in Figure~\ref{fig:StackVC-Example}. When plugging in $\price^* = (13,\infty)$, no vertex $v_i$ with $i < 7$ satisfies the \ref{eq:cover_condition}. Hence, the respective other part of the bipartition constitutes a \ac{MWVC} of $\pathGraph_{\preceq v_i}$. Vertex $v_7$ fulfills the \ref{eq:cover_condition} and thus the set containing $v_7$ and every second of its predecessors constitutes a \ac{MWVC} of $\pathGraph_{\preceq v_7}$. If this cover is chosen, it suffices to treat the remaining path separately.
	
	\section{Bounds for Priceable Vertices}
	\label{sec:cutting+bounds}
	As we motivated before, the leader generally wants to include a specific selection of priceable vertices in the follower's solution. In the following, we therefore derive bounds for each priceable vertex $\p_i$ telling us when a priceable vertex will be in- or excluded from a \ac{MWVC}. First, we give the definition, then explain the derivation, and finally, show how the bounds can be computed efficiently.
	
	\begin{definition}[Lower and upper Bounds]
		\label{def:BoundIntuition}
		Let $\p_i$ be a priceable vertex of $\pathGraph$. Given prices $\price_1, \dots, \price_{i-1}$ for all preceding priceable vertices, we denote the lower and upper bound of $\p_i$ as
		\begin{enumerate}[(i)]
			\item{$\olb[i](\price_1, \dots, \price_{i-1})$} -- the maximum price for $\p_i$ such that $\p_i$ is included in any \ac{MWVC} of $\pathGraph$ independent of the prices $\price_{i+1}, \dots, \price_k$ chosen for the remaining priceable vertices $\p_{i+1}, \dots, \p_k$;
			\item{$\oub[i](\price_1, \dots, \price_{i-1})$} -- the maximum price $\price_i$ such that $p_i$ can still be included in a \ac{MWVC} of $\pathGraph$ by setting the prices of $\p_{i+1}, \dots, p_k$ accordingly.
		\end{enumerate}		
		
		Note that both bounds, $\olb[i](\price_1, \dots, \price_{i-1})$ and $\oub[i](\price_1, \dots, \price_{i-1})$, are functions from $\mathbb{R}^{i-1} \to \mathbb{R}$. However, if the input $\price_1, \dots, \price_{i-1}$ is fixed, we write $\olb[i]$ and $\oub[i]$ for the sake of convenience.
	\end{definition}
	
	The rationale for computing $\olb[i]$ builds on Lemma~\ref{lem:in_cover}. We search for the maximum price $\price_i = \olb[i]$ such that one node in $u \in \Partition{\p_i}$ between $\p_i$ and $\p_{i+1}$, \ie $\p_i \preceq u \prec \p_{i+1}$, satisfies the \ref{eq:cover_condition}.
	We further ensure that when setting $\price_i = \olb[i]$, no vertex $w$ from $\pathGraph \setminus \Partition{\p_i}$ between $p_i$ and $u$ satisfies the \textit{strict} \ref{eq:cover_condition}.
	Moreover, we cut the instance whenever the strict \ref{eq:cover_condition} holds at some predecessor of $p_i$. This guarantees the existence of a \ac{MWVC} including every $v \in \Partition{p_i}$ with $p_i \preceq v \preceq u$.
	
	To derive $\oub[i]$,  we use a similar argument as above for $\olb[i]$. However, this time, we want to find the maximum price such that no $w \in  \pathGraph \setminus \Partition{\p_i}$ between $p_i$ and $p_{i+1}$ fulfills the strict \ref{eq:cover_condition}.
	Note that the existence of such a node $w$ would force $\p_i$ out of any \ac{MWVC}.
	
	On the same occasion, we make the following observation:
	\begin{observation}
		\label{obs:exclude_by_oub}
		Setting $\price_i = \oub[i]$ suffices to obtain an optimal cover excluding $\p_i$. It is thereby the minimum price guaranteeing the existence of a \ac{MWVC} excluding $p_i$. We discuss later in Section~\ref{subsec:Interaction_OpposingPartitions} why this may be desirable.
		\\
		One important special case is the last priceable vertex $\p_k$. For this vertex, the lower and upper bound always coincide, \ie $\olb[k]= \oub[k]$, by definition.
	\end{observation}
	
	\Cref{alg:ComputeBounds} shows how the above principles can be implemented to compute the bounds efficiently. The main idea is to evaluate the \ref{eq:cover_condition} for every vertex $v$ and compute the difference or \textit{gap} between the weights of the vertices from the two parts of the bipartition up to $v$:
	\begin{definition}
		\label{def:gap}
		Let $\p_i \in P$ be a priceable vertex and let $\price_1, \dots, \price_{i-1}$ be the prices of the priceable predecessors of $\p_i$. \Wlog no vertex $u \prec \p_i$ satisfies the strict \ref{eq:cover_condition} (otherwise use $\pathGraph = \pathGraph_{u\prec}$ for the following definition). For a vertex $v$ with $\p_i \preceq v \prec \p_{i+1}$ (or $\p_i \preceq v$ if $i=k$), we define $\gap(v)$ as the difference in weight of the two parts of the bipartition of $\pathGraph_{\preceq v}$. More concretely,
		\begin{equation}
			\label{eq:gap_definition}
			\gap(v) \coloneqq \begin{cases}
			\omega(\pathGraph_{\preceq v} \cap \Partition{v}) - \omega(\pathGraph_{\preceq v} \setminus \Partition{v}) & \text{if } v \in \Partition{v},\\
			\omega(\pathGraph_{\preceq v} \setminus \Partition{v}) - \omega(\pathGraph_{\preceq v} \cap \Partition{v}) & \text{if } v \notin \Partition{v}
			\end{cases}
		\end{equation}
		where we use $\omega(p_j) = \price_j$ for $j<i$ and $\omega(p_i) = 0$.\footnote{This definition differs from the one in the conference version \cite{StackVCP_SAGT23}. The definition of $\gap(v)$ using the absolute value neglects that the above differences can be negative.}
	\end{definition}

	\begin{algorithm}[tbhp]
	\algorithmInit
	\Fn{\ComputeBounds{$\pathGraph_{\prec p_{i+1}}, \price_1, \dots, \price_{i-1}$}}{	
		\KwIn{$\pathGraph_{\prec p_{i+1}}= (v_1, \dots, v_\ell)$, prices $\price_1, \dots, \price_{i-1}$ for $\p_1, \dots, \p_{i-1}$.}
		\KwOut{Lower bound $\olb[i]$ and upper bound $\oub[i]$ for $\p_i$.}
		Cut instance at $v \prec p_i$ whenever the strict cover condition is fulfilled and repeat with $\pathGraph_{v \prec}$\; 
		$\omega(\p_i)$ = $0$\;
		Initialize $\beta[i,\mbox{pred}(p_i)]$ = $\infty$, $\alpha[i,\mbox{pred}(\mbox{pred}(p_i))]$ = $0$ \;
		\For(\tcp*[f]{If $p_i = p_k$, go through $\p_i \preceq v \preceq v_n$}){$v$ \With $\p_i \preceq v \prec{p_{i+1}}$}{
			$\gap(v)$ = $\begin{cases}
			\omega(\pathGraph_{\preceq v} \cap \Partition{v}) - \omega(\pathGraph_{\preceq v} \setminus \Partition{v}) & \text{if } v \in \Partition{v},\\
			\omega(\pathGraph_{\preceq v} \setminus \Partition{v}) - \omega(\pathGraph_{\preceq v} \cap \Partition{v} & \text{if } v \notin \Partition{v})
			\end{cases}$ \;
			\uIf{$v \in \Partition{\p_i}$}{
				$\alpha[i,v]$ = $\max\{\gap(v),\ \alpha[i, \mbox{pred}(\mbox{pred}(v))]\}$\;
				$\beta[i,v]$ = $\beta[i,\ \mbox{pred}(v)]$\;
			}
			\Else{
				$\alpha[i,v]$ = $\alpha[i,\   \mbox{pred}(v)]$\;
				$\beta[i,v]$ = $\min\{\gap(v), \alpha[i,\ \mbox{pred}(\mbox{pred}(v))]\}$
			}
			\If{$\alpha[i,v] > \beta[i,v]$}{
				\uIf(\tcp*[f]{include $p_i$ by cutting after $v$}){$v \in \Partition{\p_i}$}{
					$\olb[i]$ = $\oub[i] $ = $ \beta[i,v]$}
				\Else(\tcp*[f]{include $p_i$ by cutting after $u$ with $\alpha[i,v] = \gap(u)$}){
					$\olb[i]$ = $\oub[i] $ = $ \alpha[i,v]$}		
				\KwRet{$\olb[i]$, $\oub[i]$}
			}
		}
		\uIf{$i \neq k$}{
			$\olb[i]$ = $\alpha[i,v_{\ell}]$, $\oub[i]$ = $\beta[i,v_{\ell}]$\;
		}
		\Else(\tcp*[f]{Special treatment for $p_k$ due to end of path}){
			\uIf(\tcp*[f]{cut after $u$ with $\alpha[i,u] = \gap(u)$}){$v_{\ell} \in \Partition{\p_i}$}{
					$\olb[i]$ = $\oub[i]$ = $\alpha[i,v_{\ell}]$}
			\Else(\tcp*[f]{include $\p_i$ by excluding $v_{\ell}$}){
					$\olb[i]$ = $\oub[i]$ = $\beta[i,v_{\ell}]$}
		}
		\KwRet{$\olb[i]$, $\oub[i]$}
	}
	\setstretch{1}
	\caption{Algorithm to compute $\olb[i]$ and $\oub[i]$.}
	\label{alg:ComputeBounds}
\end{algorithm}
	
	Intuitively speaking, the gap tells us how high we can set $\price_i$ until $v$ triggers the (weak) \ref{eq:cover_condition}. Consequently, we want to find a node $u \in \Partition{\p_i}$ with a very large gap such that no $w \in \pathGraph \setminus \Partition{\p_i}$ between $\p_i$ and $u$ has a strictly smaller gap. Algorithm~\ref{alg:ComputeBounds} essentially does exactly this.	
	
	\begin{lemma}
		\label{lem:correctness_boundAlgo}
		\Cref{alg:ComputeBounds} computes the bounds $\olb[i]$ and $\oub[i]$ as introduced in Definition~\ref{def:BoundIntuition}.
	\end{lemma}
	\begin{proof}
			The proof follows from Lemma~\ref{lem:in_cover}, Corollary~\ref{cor:cutting}, and Corollary~\ref{cor:VC} which apply because we cut the path whenever the strict \ref{eq:cover_condition} condition is fulfilled at some node $v$.
			
			Clearly, \Cref{alg:ComputeBounds} computes the correct optimal bounds for any $\p_i \neq \p_k$ if the lower- and upper bound do not intersect.
			In the following, we prove the correctness if the bounds intersect, \ie $\alpha[i,v] \geq \beta[i,v]$, or if the end of $\pathGraph$ is reached.	
			
			\medskip
			
			To show that there is a unique optimal price for $\p_i$ if the bounds $\alpha[i,v]$ and $\beta[i,v]$ intersect at some vertex $v$, we distinguish whether $v$ lies in $\Partition{\p_i}$.
			\begin{description}
				\item [$v \in \Partition{\p_i}$:] We know that $\gap(v)>\alpha[i,\pred(v)]$, as the bounds did not intersect before. For $\price_i = \beta[i,v]$, no $w \in \pathGraph \setminus \Partition{\p_i}$ between $\p_i$ and $v$ fulfills the strict \ref{eq:cover_condition}. Conversely, $v$ satisfies the strict \ref{eq:cover_condition} under $\price_i = \beta[i,v] < \alpha[i,v] = \gap(v)$. Hence, for any $\price_i > \beta[i,v]$, there is no \ac{MWVC} containing $\p_i$ whereas for $p_i \leq \beta[i,v]$ vertex $\p_i$ is contained in every \ac{MWVC} of $\pathGraph$. Hence, $\olb[i] = \oub[i] = \beta[i,v]$ is the unique optimal choice given that $p_i$ should be part of the follower's solution.
				\item [$v \notin \Partition{\p_i}$:] As above, we have $\gap(v) > \beta[i,\pred(v)]$, since the bounds did not intersect before. For price $\price_i > \alpha[i,v] > \beta[i,v]$ there is no $u \in \Partition{v}$ between $p_i$ and $v$ satisfying the \ref{eq:cover_condition}. Moreover, at $v$ the \ref{eq:cover_condition} holds with strict inequality.\footnote{It is possible that the \ref{eq:cover_condition} is already satisfied for a vertex $v' \in \Partition{v}$ but than the same argument holds with $v'$ instead $v$.} Hence, by Lemma~\ref{lem:in_cover} all \acp{MWVC} of $\pathGraph$ include $v$ and by Corollary~\ref{cor:VC} no such cover contains $\p_i$. For $\price_i = \alpha[i,v]$ vertex $\p_i$ is contained in a \ac{MWVC} by cutting the instance after a node $u$ with $\p_i \preceq u \prec v$ and $\alpha[i,v] = \gap(u)$.	
			\end{description}
			
			\medskip
			
			If the bounds for computing $p_k$ do not coincide, we have to distinguish whether the last vertex $v_n$ of $\pathGraph$ is in $\Partition{\p_k}$.
			If it is, we need to find a vertex $u \in \Partition{\p_k}$ which fulfills the \ref{eq:cover_condition} where each vertex $w \in \Partition{\p_k}$ between $\p_k$ and $u$ does not. So this is basically the same setting as for computing $\olb[i]$ as before.
			
			If $v_n \notin \Partition{\p_k}$, we could set the price $\price_k$ to the highest value such that no vertex between $\p_k$ and $v_n$ (including $v_n$) does not fulfill the strict \ref{eq:cover_condition}. Since $v_n$ itself does not fulfill the strict \ref{eq:cover_condition}, there is a \ac{MWVC} of $\pathGraph$ which does not include $v_n$ and thus includes $\p_k$. Note, it is necessary that $v_n$ is the last vertex of the path to apply this argument.
		\end{proof}
	
	For the instance in Figure~\ref{fig:StackVC-Example}, we obtain the bounds $\olb[1] = 5$ and $\oub[1]=13$. The bounds for the second priceable vertex $\p_2$ depend on the choice for $\price_1$.
	For example, for $\pi_1 = \olb[1]$, we obtain the bounds $\olb[2]=1$, and $\oub[2] = 3$, while for $\pi_1=\oub[1]$ we obtain  the bounds $\olb[2]=9$, and $\oub[2] = 11$.
	
	\section{Interactions between Priceable Vertices}
	\label{sec:effect}
	As we motivated earlier, the mutual influence of priceable vertices is the main challenge of solving \ac{StackVC} on a path. Nevertheless, in the previous section, we computed the bounds $\olb[i]$ and $\oub[i]$ of each priceable vertex under the assumption that the prices of all preceding priceable vertices are fixed. In this section, we shed light on the interaction between priceable vertices and justify why the computation of bounds can be done in the way we explained before.
	
	To do so, we first show that choosing prices outside the bounds cannot be beneficial (cf.\ Lemma~\ref{lem:beyond_bounds}). 
	Moreover, we show that the price of $\p_i$ within the bounds does not affect the bounds for vertices $\p_{i+2}, \dots, \p_k$ if $\price_{i+1}$ is set accordingly (cf.\ Lemma~\ref{lem:effectRange}).
	Then, we investigate how the price $\price_i$ of one priceable vertex influences the bounds -- and thus ultimately also the potential revenue -- of the next priceable vertex $p_{i+1}$. We distinguish the membership of $p_{i+1}$ in $\Partition{p_i}$ (Section~\ref{subsec:Interaction_SamePartition} and \ref{subsec:AppraisingRevenue}).
	Further, we observe that there is an optimal price setting, where all prices $\price_i$ are set to their lower or upper bound \wrt $\price_1, \dots, \price_{i-1}$ (cf.\ Section~\ref{subsec:options}).
	
	\begin{lemma}
		\label{lem:beyond_bounds}
		Let $\price_1, \dots, \price_{i-1}$ be fixed, and $\olb[i]$ and $\oub[i]$ be the optimal bounds computed by \Cref{alg:ComputeBounds} for $\p_i$. The following statements hold:
		\begin{enumerate}
			\item[(i)] Choosing $\price_i < \olb[i]$ reduces the leader's maximum possible revenue on $\pathGraph$ compared to $\price_i = \olb[i]$.
			\item[(ii)] Choosing $\price_i > \oub[i]$ cannot improve the leader's maximum possible revenue  on $\pathGraph$ compared to $\price_i = \oub[i]$.
		\end{enumerate}
	\end{lemma}
	
	\begin{proof}
			The statement follows from the definition of the bounds (Definition~\ref{def:BoundIntuition}):
			\begin{enumerate}
				\item [(i)] If $\price_i < \olb[i]$, there exists a \ac{MWVC} including $\p_i$ by definition. Hence, we can cut the instance after $\p_i$ and solve the subinstances $\pathGraph_{\preceq \p_i }$ and $\pathGraph_{\p_i \prec}$ separately. Since $\price_i$ does not affect $\pathGraph_{\p_i \prec}$, we can choose the same optimal subcover as for $\price_i < \olb[i]$.
				
				Since the leader obtains $\price_i$ as revenue for $\p_i$ if $\price_i < \olb[i]$ and if $\price_i = \olb[i]$, setting $\price_i = \olb[i]$ is strictly better.
				\item [(ii)]
				If $\price_i > \oub[i]$, vertex $\p_i$ is not included in any \ac{MWVC} by definition. Thus, every \ac{MWVC} has to include $\pred(\p_i)$ and $\psucc(\p_i)$.
				We can therefore cut the instance before $\pred(\p_i)$ and after $\psucc(\p_i)$ and find a \ac{MWVC} on $\pathGraph_{\prec \pred(\p_i)}$ and $\pathGraph_{\psucc(\p_i) \prec}$ separately.
				Note that neither of these subcovers depends on the price $\price_i$ of $\p_i$.
				
				It remains to show that for $\price_i = \oub[i]$, there exists a \ac{MWVC} that does not include $\p_i$. To see this, note that either the bounds coincide (\ie $\olb[i] = \oub[i]$), or there is a vertex $w \in \pathGraph_{\prec \p_{i-1} \setminus \Partition{\p_i}}$ with $\gap(w)= \oub[i]$. 
				
				First, assume that the bounds match at some vertex $v$. If $v \in \Partition{\p_i}$, some $w$ between $\p_i$ and $v$ fulfills the \ref{eq:cover_condition}. That is, $w \notin \Partition{\p_i}$ is the first vertex with $\gap(w) = \beta[i,v]$. Conversely, no $u$ between $\p_i$ and $w$ satisfies the \ref{eq:cover_condition}. Hence, by \ref{eq:cover_condition}, there is a \ac{MWVC} of $\pathGraph$ containing $\Partition{w} = \pathGraph_{\preceq w} \setminus \Partition{\p_i}$ and thus excluding $\p_i$.
				Otherwise, if $v \notin \Partition{\p_i}$, it holds $\price_i = \oub[i] = \olb[i] = \alpha[i,\pred(v)]$. Therefore, no $u$ between $p_i$ and $v$ satisfies the strict \ref{eq:cover_condition} and the (weak) \ref{eq:cover_condition} holds at $v$. Thus, by Lemma~\ref{lem:in_cover} and Corollary~\ref{cor:cutting}, there is a \ac{MWVC} including $v$ and excluding $\p_i$.
				
				In case there is a vertex $w \in \pathGraph_{\prec \p_i} \setminus \Partition{\p_i}$ -- \WLOG we assume $w$ to be the first such vertex -- no vertex between $\p_i$ and $w$ satisfies the strict \ref{eq:cover_condition} whereas the (weak) \ref{eq:cover_condition} is fulfilled at vertex $w$. Following the same lines as above, there is a \ac{MWVC} excluding $\p_i$. 
			\end{enumerate}
		\vspace{-5ex}
		\end{proof}
	\noindent
	Next, we show that choosing a price \textit{strictly} between the bounds is not beneficial as well. Therefore, we show how $\price_i$ affects the bounds for the following priceable vertices. We use this later to narrow down the possible combinations of an optimal pricing scheme.
	
	\begin{lemma}
		\label{lem:effectRange}
		Let $\p_i$, $i<k$, be a priceable vertex of $\pathGraph$ and the prices $\price_1, \dots, \price_{i-1}$ of all preceding priceable vertices be fixed. We compare two prices $\price_i'$ and $\price_i''$ for $\p_i$  in the interval $[\olb[i], \oub[i]\,]$. Furthermore, let $x \in \mathbb{R}$ denote the difference between $\pi_i''$ and $\pi_i'$, \ie $\price_i'' = \price_i' + x$. The following statements hold:
		\begin{enumerate}[label=(\roman*),leftmargin=*]
			\item \label{enum:lemmaEffectRange_OnePart} If $\p_{i+1} \in \Partition{p_i}$, we get $\olb[i](\price_1, \mydots, \price_{i-1}, \price_i') = \olb[i](\price_1, \mydots, \price_{i-1}, \price_i'') - x$ and $\oub[i](\price_1, \mydots, \price_{i-1}, \price_i') = \oub[i](\price_1, \mydots, \price_{i-1}, \price_i'') -x$.
			\item \label{enum:lemmaEffectRange_OppPart} If $\p_{i+1} \notin \Partition{p_i}$, we get $\olb[i](\price_1, \mydots, \price_{i-1}, \price_i') = \olb[i](\price_1, \mydots, \price_{i-1}, \price_i'') + x$ and $\oub[i](\price_1, \mydots, \price_{i-1}, \price_i') = \oub[i](\price_1, \mydots, \price_{i-1}, \price_i'') + x$.
			\item \label{enum:lemmaEffectRange_implications}
			Independent of whether $\p_{i+1} \in \Partition{\p_i}$, the bounds for $\p_{i+2}, \dots, \p_{k}$ remain unchanged when $\price_{i+1}$ maintains the same distance to its bounds under both prices for $\price_i'$ and $\price_i''$ for $\p_i$.
			More concretely, if the price for $\p_{i+1}$ is set to $\price_{i+1}' = \olb[i](\price_1, \mydots, \price_{i-1}, \price_i') + y$ or $\price_{i+1}'' = \olb[i](\price_1, \mydots, \price_{i-1}, \price_i'') + y$ respectively, then $\olb[j](\price') = \olb[j](\price'')$ for all $j \geq i+2$ where $\price'_{\ell} = \price''_{\ell}$ for all $\ell \in \{1,\dots,i-1,i+1,\dots,j-1\}$.
		\end{enumerate}
	\end{lemma}
	
	\noindent
	To see that the statement holds, consider the definition of the bounds (Definition~\ref{def:BoundIntuition}) and their connection to the gap (Definition~\ref{def:gap}):
	
	\begin{proof}
			If $\price_i' = \price_i''$, the statements follow immediately. Thus, we consider the case where $\price_i' \neq \price_i''$. First, note that in this case $\olb[i] < \oub[i]$. Since we choose $\price_i'$ and $\price_i''$ within these bounds, the strict \ref{eq:cover_condition} is not fulfilled for any $v$ between $\p_i$ and $\p_{i+1}$, thus we do not cut the instance in between. \Wlog let $\price_i' < \price_i''$.
			\begin{enumerate}[(i)]
				\item 
				Looking at Definition~\ref{def:gap}, we observe that increasing $\price_i'$ by $x$ decreases $\gap(v)$ for all $v$ with $p_{i+1} \preceq v$ by the same amount and thus also $\olb[i+1]$ and $\oub[i+1]$.
				\item Here, basically the same argument holds as for \ref{enum:lemmaEffectRange_OnePart}. But since $\p_i \notin \Partition{p_i}$, we observe that increasing $\price_i$ increases $\gap(v)$ (instead of decreasing) for all $v$ with $p_{i+1} \preceq v$.
				\item If $\p_{i+1} \in \Partition{p_i}$, as we discussed in \ref{enum:lemmaEffectRange_OnePart}, the lower bound of $\p_{i+1}$ decreases by the same amount that $\price_i$ increases. When computing $\gap(v)$ for $ v \succ \p_{i+1}$ both $\p_i$ and $\p_{i+1}$ appear in the same part of the difference and since
				\begin{align*}
					\price_i'' + \price_{i+1}''
					= (\price_i'+x) + (\olb[i](\price_1, \dots, \price_{i-1}, \price_i')-x + y) = \price_i' + \price_{i+1}',
				\end{align*}
				we obtain the same values for $\gap(v)$, $p_{i+2} \preceq v$. Consequently, the bounds $\olb[j]$ and $\oub[j]$ with $j\geq i+2$ under both prices $\price'$ and $\price''$ also coincide.
				
				If $\p_{i+1} \notin \Partition{p_i}$, we obtain again the same values for $\gap(v)$. This is because the prices appear on different sides of the difference in Definition~\ref{def:gap} and because
				\begin{align*}
					\price_i'' - \price_{i+1}'' = (\price_i'+x) - (\olb[i](\price_1, \dots, \price_{i-1}, \price_i') + x + y) = \price_i' - \price_{i+1}'.
				\end{align*}
				Hence, the bounds for $\p_{i+2}, \dots, \p_{k}$ under $\price'$ and $\price''$ are also identical when $\p_{i+1}$ and $\p_i$ are on different sides of the bipartition.
			\end{enumerate}
		\vspace{-5ex}
		\end{proof}	
	\noindent
	To find the optimal prices from the leader's perspective, we must consider the effect of pricing a vertex within the described bounds. Therefore, we 
	now analyze the implications of choosing different prices within the bounds. By doing so, we are able to boil down the choices of each $\p_i$ to three options.
	In \Cref{subsec:AppraisingRevenue}, we then describe the revenue corresponding to these options. Finally, in \Cref{subsec:comparing_options} we show how to compare the options by their revenues.
	
	\subsection{Options}
	\label{subsec:options}
	In the following, we narrow the reasonable pricing schemes, \ie a combination of prices, for the leader. For this purpose, we distinguish whether $\p_{i+1}$ is in the same or in the opposing part of the bipartition as $\p_i$.
	
	\paragraph{Same part.}
	\label{subsec:Interaction_SamePartition}
	We begin with $\p_{i+1} \in \Partition{p_i}$. Observe that $\p_i$ can be included directly by choosing $\price_i = \olb[i]$. According to Lemma~\ref{lem:beyond_bounds}, any price below this is disadvantageous. Any $\price_i$ in the interval $(\olb[i], \oub[i]\,]$ has the same effect, namely that $\p_i$ is included if and only if $p_{i+1}$ is included as well. Looking at Lemma~\ref{lem:effectRange}, we observe that increasing $\price_i$ decreases $\olb[i+1]$ and $\oub[i+1]$ by the same amount that we increase $\price_i$. Since we must include $\p_{i+1}$ to include $\p_i$ for any $\price_i > \olb[i]$, we can \textit{shift} any excess revenue beyond $\olb[i]$ from $\p_i$ to $\p_{i+1}$ by setting $\price_i = \olb[i]$.
	
	\begin{lemma}[Options for priceable vertices in same part]
		\label{lem:options_OnePart}
		Let $i <k$ and $\p_{i+1} \in \Partition{\p_i}$, then the following statements hold
		\begin{enumerate}
			\item [\namedlabel{enum:CorollaryOptPrices1}{(i)}] Given that $\p_i$ should be contained in the follower's solution, $\price_i = \olb[i]$ is an optimal price for $\p_i$.
			\item [\namedlabel{enum:CorollaryOptPrices2}{(ii)}] Given that $\p_i$ should not be contained in the follower's solution, $\price_i = \oub[i]$ is an optimal price for $\p_i$.
		\end{enumerate}
	\end{lemma}
	\begin{proof}
		To see assertion~(i), remember that we cannot include $\p_i$ for a higher price without including $\p_{i+1}$. If we include $\p_{i+1}$, the revenue shifts from $\p_i$ to $\p_{i+1}$ (cf.\ Lemma~\ref{lem:effectRange}\ref{enum:lemmaEffectRange_OnePart}), there is no difference between setting $\price_i$ to $\olb[i]$ or to any value between $\olb[i]$ and $\oub[i]$.
		
		To see assertion~(ii), recall that by Definition~\ref{def:BoundIntuition}, $\oub[i]$ suffices to exclude $\p_i$. Furthermore, by Lemma~\ref{lem:beyond_bounds} any price beyond it induces no change for the following vertices.
	\end{proof}
	
	\noindent
	In summary, we showed that whenever two consecutive priceable vertices $p_i$ and $p_{i+1}$ lie on the same side of the bipartition, $\olb[i]$ is an optimal choice for $\price_i$ unless we want to exclude $\p_i$. In that case, $\price_i = \oub[i]$ is optimal. We now explore why excluding $\p_i$ might be desirable.
	
	\paragraph{Opposing part.}
	\label{subsec:Interaction_OpposingPartitions}
	We change the setting and now assume that $\p_{i+1} \notin \Partition{\p_i}$. Furthermore, we consider prices $\price_1, \dots, \price_{i-1}$ as fixed. Looking at the gap of $\price_i$ (cf. Definition~\ref{def:gap}) it is very apparent that increasing $\price_i$ also increases the gap at any successor of $\p_{i+1}$. This directly translates to larger bounds and thus ultimately a higher price for $\p_{i+1}$. This is rather intuitive since the vertices lie on different sides of the bipartition. Again, this effect is limited to $\price_i \in [\olb[i],\,\oub[i]\,]$.
	
	Now, recall that for any $\olb[i] < \price_i \leq \oub[i]$, we must rely on $\p_{i+1}$ to include $\p_i$. This follows from the definition of $\olb[i]$ (cf.\ Definition~\ref{def:BoundIntuition}). In particular, $\p_{i+1}$ must be excluded such that there can be a \ac{MWVC} containing $\pathGraph_{\prec \p_{i+1}} \cap \Partition{p_i}$. Due to Lemma~\ref{lem:effectRange}, choosing any price $\price_i$ in the interval $(\olb[i], \oub[i]\,]$ is never better than setting $\price_i = \oub[i]$. To see this, recall that if $\p_{i+1}$ is not part of the follower's solution, $p_i$ is for any $\price_i \leq \olb[i]$ per definition.
	
	Summarizing above observations, we obtain three reasonable choices for $\p_i$ if $\p_{i+1} \notin \Partition{p_i}$:
	\begin{lemma}[Options for priceable vertices in opposing parts]
		\label{lem:Option_twoPartition}
		\hspace{0pt}
		\begin{description}
			\item [\namedlabel{desc:option1}{$\Option_1$}{$[$}$i${$]$}:] Renounce $p_i$ to increase the revenue on $p_{i+1}$ (by setting $\price_i = \oub[i]$).
			\item [\namedlabel{desc:option2}{$\Option_2$}{$[$}$i${$]$}:] Include $p_i$ in a \ac{MWVC} of $\pathGraph_{\prec v}, \, \p_i \preceq v \prec \p_{i+1}$ (by setting $\price_i = \olb[i]$).
			\item [\namedlabel{desc:option3}{$\Option_3$}{$[$}$i${$]$}:] Include $p_i$ by excluding $p_{i+1}$ (by setting $\price_i = \oub[i]$ and $\price_{i+1} = \oub[i+1]$).
		\end{description} 	
	\end{lemma}
	\begin{proof}
		The proof is a direct consequence of Lemma~\ref{lem:beyond_bounds}, Lemma~\ref{lem:effectRange} and our considerations above.
	\end{proof}
	
	\noindent
	Note that option~\ref{desc:option3}[$i$] already fixes $\price_{i+1}$ and thus, there are no options for $\p_{i+1}$ in this case. Consequently, it may still be that $\p_{i+1}$ is excluded even if $\p_{i+2} \in \Partition{\p_{i+1}}$. Following our deliberations from earlier option~\ref{desc:option2}[$i$] is the only reasonable choice when $\p_{i+1} \in \Partition{\p_i}$. The following observation captures this formally.
	
	\begin{observation}
		For $\p_i$ with $\p_{i+1} \in \Partition{p_i}$, the options condense to $\text{\ref{desc:option2}}[i]$. In particular, choosing $\p_i = \oub[i] > \olb[i]$ can only be reasonable if $\p_{i-1} \notin \Partition{\p_i}$ and $\text{\ref{desc:option3}}[i-1]$ is chosen.
	\end{observation}
	For simplicity, we refer to the options $\text{\ref{desc:option1}}[i]$ and $\text{\ref{desc:option3}}[i]$ regardless of whether $\p_{i+1} \in \Partition{p_i}$ or not. Yet, in case $\p_{i+1} \in \Partition{p_i}$, these options are irrelevant.
	
	In the following, we explain how to appraise the revenue corresponding to the options $\text{\ref{desc:option1}}[i]$~--~$\text{\ref{desc:option3}}[i]$. Afterward, in Section~\ref{subsec:comparing_options}, we elucidate how to evaluate them from an algorithmic perspective.
	
	\subsection{Appraising Revenue}
	\label{subsec:AppraisingRevenue}
	Since the bounds and thus ultimately the potential revenue for priceable vertex $\p_i$ depends on the prices $\price_1, \dots, \price_{i-1}$ of all priceable vertices preceding it, we define a \textit{benchmark price} for all these vertices $\p_j$ with $j < i$ when evaluating $\price_i$. Based on this benchmark, we anticipate the change in revenue for all following priceable vertices $\p_{i+1}, \dots, \p_{\ell}$ when evaluating $p_i$. Independent from our choice of $\price_i$, we continue evaluating the options for $\p_{i+1}$ as if $\price_i$ matched the benchmark.  
	
	In general, any benchmark between $\olb[j]$ and $\oub[j]$ works. Yet, as we see later, assuming $\price_j = \olb[j]$ for any $j < i$ is the most convenient choice. We define as $\mathcal{R}_{i}$ the maximum revenue that the leader can obtain from vertices $\p_i, \dots, \p_k$ when $\price_j = \olb[j]$ for all $j < i$. Similarly, we define $\mathcal{R}_{i}(\Option_\ell)$ as $\mathcal{R}_{i}$ when fixing one of the options $\text{\ref{desc:option1}}[i]$~--~$\text{\ref{desc:option3}}[i]$ for $\p_i$.
	
	Given that, we can compute the leader's revenue from $p_i$ onward as described in the following two lemmata.
	
	\begin{lemma}[Revenues for priceable vertices in same part]
		\label{lem:RevenueEstimation_onePart}
		Let $p_i$ be a priceable vertex of $\pathGraph$ with $i \leq (k-1)$ and $p_{i+1} \in \Partition{\p_i}$. Assuming $\price_j = \olb[j]$ for every $j < i$; the maximum revenue that the leader can obtain on $\pathGraph_{\p_i \prec}$ always corresponds to choosing $\text{\ref{desc:option2}}[i]$ introduced above:
		\begin{alignat}{4}
			\mathcal{R}_{i}(\text{\ref{desc:option2}}) &= \olb[i] + \mathcal{R}_{i \Shortplus 1}.
			\span \span
			\label{eq:RevenueO2'}\tag{R2'}
		\end{alignat}
		Note that we do not state how to the appraise the revenue of $\text{\ref{desc:option1}}[i]$ and $\text{\ref{desc:option3}}[i]$ since these options are generally inferior.
	\end{lemma}
	\begin{proof}
		This follows directly using Lemma~\ref{lem:effectRange} and the considerations in Section~\ref{subsec:AppraisingRevenue} in case $p_{i+1}$ lies in the same partition as $\p_i$.
	\end{proof}
	
	\begin{lemma}[Revenues for priceable vertices in opposing parts]
		\label{lem:RevenueEstimation_oppPart}
		Let $p_i$ be a priceable vertex of $\pathGraph$ with $i \leq (k-1)$ and $p_{i+1} \notin \Partition{\p_i}$. Assuming $\price_j = \olb[j]$ for every $j < i$; the maximum revenue that the leader can obtain on $\pathGraph_{\p_i \prec}$ given the options introduced above is:
		\begin{alignat}{4}
			\mathcal{R}_{i}(\text{\ref{desc:option1}}) &= (\oub[i] - \olb[i]) + \mathcal{R}_{i \Shortplus 1},
			\span \span
			\label{eq:revenueO1}\tag{R1}\\
			\mathcal{R}_{i}(\text{\ref{desc:option2}}) &= \olb[i] + \mathcal{R}_{i \Shortplus 1},
			\span \span
			\label{eq:RevenueO2}\tag{R2}\\
			\mathcal{R}_{i}(\text{\ref{desc:option3}}) &=
			\begin{cases}
				\oub[i] + (\oub[i \Shortplus 1] - \olb[i \Shortplus 1]) + \mathcal{R}_{i \Shortplus 2} & \text{if } p_{i \Shortplus 2} \notin \Partition{p_{i \Shortplus 1}},\\
				\oub[i] - (\oub[i \Shortplus 1] - \olb[i \Shortplus 1]) + \mathcal{R}_{i \Shortplus 2} &\text{if } \; p_{i \Shortplus 2} \in \Partition{p_{i \Shortplus 1}}, \\
				\oub[i] &\text{if } i+2 >k.
			\end{cases}
			\label{eq:RevenueO3}\tag{R3}
		\end{alignat}
	\end{lemma}
	
	\medskip
	\begin{proof}
		We begin with $\text{\ref{desc:option1}}$ which corresponds to renouncing $\p_i$ with the goal of increasing the revenue on $\p_{i+1}$. By Observation~\ref{obs:exclude_by_oub}, $\oub[i]$ suffices to obtain a \ac{MWVC} of $\pathGraph$ excluding $\p_i$. Furthermore, by Lemma~\ref{lem:beyond_bounds}, exceeding the upper cannot be beneficial compared to $\price_i = \oub[i]$. Thus, this price also allows for the maximum increase of $\price_{i+1}$.
		In accordance with Lemma~\ref{lem:effectRange}, $\p_i$ does not influence any of the vertices $\p_{i+2}, \dots, \p_k$ given that $\price_{i+1}$ remains relative to its (now shifted) bounds.
		Clearly, $\text{\ref{desc:option1}}[i]$ yields no revenue on $\p_i$ but increases the revenue on $\p_{i+1}$ compared to our benchmark $\price_i = \olb[i]$. Recall that $\p_i$ and $\p_{i+1}$ are on different sides of the bipartition and thus Lemma~\ref{lem:effectRange}~\ref{enum:lemmaEffectRange_OnePart} yields that the extent of this increase is precisely $\oub[i] - \olb[i]$.
		
		For the second option, we obtain the benchmark price $\price_i = \olb[i]$ as revenue. By Definition~\ref{def:BoundIntuition}, this price is optimal given that we want to include $\p_i$ in a \ac{MWVC} of some subpath $\pathGraph_v$, with $\p_i \preceq v \prec \p_{i+1}$. Since we do not deviate from the benchmark, there is no change in revenue for $\pathGraph_{\p_{i+1} \preceq}$ to be anticipated.
		
		Lastly, consider $\text{\ref{desc:option3}}[i]$. To see why the claimed prices lead to maximum revenue \wrt the desired cover, recall that by Observation~\ref{obs:exclude_by_oub}, $\price_{i+1} = \oub[i+1]$ suffices to obtain a \ac{MWVC} excluding $\p_{i+1}$. Furthermore, by Lemma~\ref{lem:beyond_bounds}, increasing $\price_{i+1}$ beyond $\oub[i+1]$ induces no change for the revenue on the remaining path. Furthermore, $\p_i$ is excluded from any \ac{MWVC} of $\pathGraph$ if $\price_i > \oub[i]$ (cf. Definition~\ref{def:BoundIntuition}). Since $\p_{i+1}$ is not included and no vertex from $\pathGraph \setminus \Partition{\p_i}$ between $\p_i$ and $\p_{i+1}$ satisfies the strict \ref{eq:cover_condition}, $\p_i$ is also included at this price.
		
		Now, to see that Equation~\eqref{eq:RevenueO3} provides the corresponding revenue based on whether $\p_{i+1}$ and $\p_{i+2}$ are in the same part of the partition, observe the following:
		In both cases, the leader obtains a revenue of $\price_i(\text{\ref{desc:option3}}) = \oub[i]$ for $\p_i$ but no revenue for $\p_{i+1}$. This leaves her to maximize the revenue on $\pathGraph_{\p_{i+2}\preceq}$. Now, if $\p_{i + 2} \in \Partition{\p_{i + 1}}$, she essentially \emph{shifts} revenue from $p_{i+2}$ to $p_{i+1}$ (cf.\ Lemma~\ref{lem:effectRange}) thereby reducing the revenue on $p_{i+2}$ compared to our benchmark. Conversely, if $\p_{i + 2} \notin \Partition{\p_{i + 1}}$ the situation is similar to $\mathcal{R}_i(\text{\ref{desc:option1}})$ and thus the effect reverses.
		If $i+2>k$, vertex $\p_{i+1}$ is the last priceable vertex on $\pathGraph$. This third case of Equation~\eqref{eq:RevenueO3} follows directly since we include $\p_i$ for price $\price_i = \oub[i]$ by excluding $\p_{i+1}$, and there is no other vertex to be influenced by the deviation of $\price_{i+1}$ from its benchmark.
	\end{proof}
	\medskip
	
	\subsection{Compare Options and Resolve Prices}
	\label{subsec:comparing_options}
	Putting everything together, we now discuss how the options introduced in Section~\ref{subsec:options} can be evaluated efficiently. Since the number of possible combinations is exponential in $k$, comparing all of them would result in an exponential blow-up. Fortunately, Lemma~\ref{lem:RevenueEstimation_onePart} and Lemma~\ref{lem:RevenueEstimation_oppPart} provide us with the means to compute the optimal choice for $\p_i$ already when we consider $\p_{i+1}$. Recall that we assume $\price_j = \olb[j]$ as the benchmark price for all $\p_j \in P$. This allows
	to compute fixed bounds $\olb[j]$, $\oub[j]$ for every priceable vertex. By anticipating any changes in revenue that arise from deviating from this benchmark, we can compare the revenue on the remaining path.
	
	With that, we can compare $\text{\ref{desc:option1}}[i]$ and $\text{\ref{desc:option2}}[i]$ right away. Since $\text{\ref{desc:option3}}[i]$ already fixes a price for the next priceable vertex $\p_{i+1}$, it can only be compared to $\text{\ref{desc:option1}}[i+1]$ and $\text{\ref{desc:option2}}[i+1]$. The important insight is that independent of which of the options is best, we find a clear optimal choice for $\price_i$. To evaluate $\p_{i+1}$, we might need to open a new branch but since there is only one left for $\p_{i}$, we end up with at most two open branches at any time. Figure \ref{fig:Branching-GeneralCase} illustrates this procedure.
	\medskip
	
	\tikzset{
	every node/.style={
		thick,
		inner xsep = 3pt,
		inner ysep = 2pt,
		text = black,
		font = {\bfseries\small},
	},
}
\begin{figure}[tbph]
	\resizebox{\textwidth}{!}{
		\centering
		\begin{tikzpicture}
			\centering
			\node[rectangle, textbox] (AGR20) {$\AGR{2}{i \Shortminus 2}$};
			\node[above = \fullVertSpacing of AGR20, anchor = south, vertex, draw = none](d0){\tikzdots};
			\node[below = 1.5cm of AGR20, outer sep = 0, shape = rectangle, minimum width= .7cm, minimum height= 0cm](h0){};
			\node[right = 1cm of AGR20.east, textbox](O21){\color{RWTH_Magenta}$\Option_2[i \Shortminus 1]$};
			\node[above = \vertSpacing of O21.north, anchor=south, textbox](O11){\color{HKS44}$\Option_1[i \Shortminus 1]$};
			\node[below = \vertSpacing of O21.south, anchor=north, textbox](O31){\color{RWTH_Orange}$\Option_3[i \Shortminus 1]$};
			\node[above = \vertSpacing of O11.north, anchor = south, vertex](p1){\normalsize $p_{i-1}$}; 
			\node[right= 1cm of O21, textbox](AGR11){$\AGR{1}{i \Shortminus 1}$};
			\node[right = 2.9cm of p1, vertex, draw = none](d1){\tikzdots};
			\node[right= \vertSpacing of AGR11.east, anchor=west](AGR21){$\AGR{2}{i \Shortminus 1}$};
			\node[right = 1cm of AGR21, textbox](O22){\color{RWTH_Magenta}$\Option_2[i]$};
			\node[above = \vertSpacing of O22, textbox](O12){\color{HKS44}$\Option_1[i]$};
			\node[below = \vertSpacing of O22, textbox](O32){\color{RWTH_Orange}$\Option_3[i]$};
			\node[above = \vertSpacing of O12.north, anchor = south, vertex, shape = circle](p2){\normalsize $p_{i}$};
			\node[right= 1cm of O22, textbox](AGR12){$\AGR{1}{i}$};
			\node[right = 2.4cm of p2, vertex, draw = none](d2){\tikzdots};
			\node[right= \vertSpacing of AGR12, textbox](AGR22){$\AGR{2}{i}$};
			\node[right = 1cm of AGR22, textbox](O23){\color{RWTH_Magenta}$\Option_2[i \Shortplus 1]$};
			\node[above = \vertSpacing of O23, textbox](O13){\color{HKS44}$\Option_1[i \Shortplus 1]$};
			\node[below = \vertSpacing of O23, textbox](O33){\color{RWTH_Orange}$\Option_3[i \Shortplus 1]$};
			\node[above = \vertSpacing of O13.north, vertex](p3){\normalsize $p_{i+1}$};
			\node[right= 1cm of O23, textbox](AGR13){$\AGR{1}{i \Shortminus 1}$};
			\node[right= 1.55cm of p3, vertex, draw = none](d3){\tikzdots};
			\node[right= 1.55cm of d3, vertex, shape = circle](p4){$p_{i+2}$};
			\node[right= 1cm of AGR13.east, outer ysep=.1cm](AGR23){\tikzdots};
			\node[below = 1.8cm of AGR23, outer sep = 0, shape = rectangle, minimum width= .7cm, minimum height= 0cm, outer ysep=.1cm](h4){\tikzdots};
			\draw[-stealth, bend right, out=-28, in=-150] (h0.north west) to (AGR21);
			\draw[-stealth] (AGR20.north east) to (O11.south west);
			\draw[-stealth] (AGR20.east) to (O21.west);
			\draw[-stealth] (AGR20.south east) to (O31.north west);
			\draw[-stealth] (O11.south east) to (AGR11.north west);
			\draw[-stealth] (O21.east) to (AGR11.west);
			\draw[-stealth, bend right, out=-20, in=-140] (O31.east) to (AGR22);
			\draw[-stealth] (AGR11) to (AGR21);
			\draw[-] (d0) to (p1);
			\draw[-] (p1) to (d1);
			\draw[-stealth] (AGR21.north east) to (O12.south west);
			\draw[-stealth] (AGR21.east) to (O22.west);
			\draw[-stealth] (AGR21.south east) to (O32.north west);
			\draw[-stealth] (O12.south east) to (AGR12.north west);
			\draw[-stealth] (O22.east) to (AGR12.west);
			\draw[-stealth, bend right, out=-20, in=-140] (O32.east) to (AGR23);
			\draw[-stealth] (AGR12) to (AGR22);
			\draw[-] (d1) to (p2);
			\draw[-] (p2) to (d2);
			\draw[-stealth] (AGR22.north east) to (O13.south west);
			\draw[-stealth] (AGR22.east) to (O23.west);
			\draw[-stealth] (AGR22.south east) to (O33.north west);
			\draw[-stealth] (O13.south east) to (AGR13.north west);
			\draw[-stealth] (O23.east) to (AGR13.west);
			%
			\draw[-stealth] (AGR13) to (AGR23);
			\draw[-] (d2) to (p3);
			\draw[-] (p3) to (d3);
			\draw[-] (d3) to (p4);
			\draw[-stealth, out=-20, in=170] (O33.east) to (h4);
		\end{tikzpicture}
	}
	\caption{A sequence of priceable vertices from alternating parts of the bipartition and their pricing options -- as introduced in Subsection~\ref{subsec:options}. At any point in time, only two options remain open.}
	\label{fig:Branching-GeneralCase}
\end{figure}
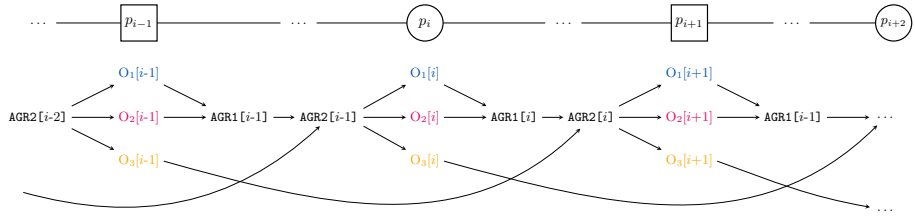

	
	\paragraph{Algorithm~\ref{alg:comparing} \texttt{CompareOptions}.}
	\SetKwFunction{ComparingOptions}{ComparingOptions}%
\SetKwFunction{ComputeOptionsEOP}{ComputeOptionsEOP}%
\begin{algorithm}[tbph]
	\algorithmInit
	\Fn{\CompareOptions{$\pathGraph, \olb, \oub$}}{	
		\KwIn{A \ac{StackVC}-instance $\pathGraph = (v_{1}, \dots, v_{n})$ with priceable vertices $p_1, \dots,p_k$. Lower and upper bounds $\olb[i]$, $\oub[i]$ for $\p_i$, assuming a benchmark price $\price_j = \olb[j]$ for every $j < i$.}
		\KwOut{An optimal choice from $\Option_1[i] - \Option_3[i]$ for each $\p_i$.}
		Initialize $\revOpt{0}{3}$ = $- \infty$ \tcp*[f]{no case distinction for $i=1$ needed}\;
		\For{$i \coloneqq 1$ \To $k$}{
			\uIf(\tcp*[f]{Compute the revenue for $\Option_1[i]$ - $\Option_3[i]$}){$i<k$}{
				\uIf{$\p_{i+1} \in \Partition{\p_i}$}{
					Compute $\revOpt{i}{1} - \revOpt{i}{3}$ with Lemma~\ref{lem:RevenueEstimation_onePart}
				}
				\Else(\tcp*[f]{$\p_{i+1} \in \pathGraph \setminus \Partition{\p_i}$}){
					Compute $\revOpt{i}{1} - \revOpt{i}{3}$ with Lemma~\ref{lem:RevenueEstimation_oppPart}
				}	
			}
			\ElseIf(\tcp*[f]{end of the path, ignore irrelevant options}){$i=k$}{
				
				$\revOpt{k}{1}$ = $\revOpt{k}{3}$ = $- \infty$ \;
				$\revOpt{k}{2}$ = $\olb[k]$ = $\oub[k]$\;
			}
			$\AGR{1}{i}$ = max $\{\revOpt{i}{1},\revOpt{i}{2}\}$ \tcp*[r]{discard inferior option} \label{line:AGR1}
			$\AGR{2}{i}$ = max $\{\AGR{1}{i}, \revOpt{i-1}{3} \}$ \tcp*[r]{discard inf.\ option for $\p_{i-1}$} \label{line:AGR2}
			backtrack optimal choice for $\p_{i-1}$ from $\AGR{2}{i}$ \label{line:backtrackChoice}
		}
		\KwRet{optimal pricing scheme}	
	}
	\setstretch{1}
	\caption{An algorithm to determine optimal pricing scheme given the benchmark prices.}
	\label{alg:comparing}
\end{algorithm}
	In the following, we describe the details of our evaluation procedure along the lines of Algorithm~\ref{alg:comparing}. The main idea is to evaluate the revenue associated to $\text{\ref{desc:option1}}[i]$ - $\text{\ref{desc:option3}}[i]$ for every $\p_i$. In case an option is contextually unavailable or generally inferior, we set the associated revenue to~$-\infty$.
	
	In line~\ref{line:AGR1}, we assess which of the comparable options at $\p_i$, namely $\text{\ref{desc:option1}}[i]$ and $\text{\ref{desc:option2}}[i]$, is superior. Given that, we can asses whether excluding $\p_{i-1}$ is beneficial or not. This is done in line~\ref{line:AGR2}. This corresponds to comparing $\text{\ref{desc:option1}}[i]$, $\text{\ref{desc:option2}}[i]$, and $\text{\ref{desc:option3}}[i-1]$. Consider Lemma~\ref{lem:RevenueEstimation_oppPart} to see why this is viable. 
	It remains to trace back which option for $\p_{i-1}$ led to this optimal revenue. We cover this in line~\ref{line:backtrackChoice}.
	
	For the instance depicted in Figure~\ref{fig:StackVC-Example}, we obtain the benchmark bounds $\olb[1]=5$, $\oub[1]=13$ for the first and $\olb[2] = \oub[2]=1$ for the second (and last) priceable vertex. We obtain $\mathcal{R}_1(\text{\ref{desc:option1}})=(13-5) + \mathcal{R}_{i+1} > 5 + \mathcal{R}_{i+1} = \mathcal{R}_1(\text{\ref{desc:option2}})$, telling us that we can discard $\text{\ref{desc:option2}}[i]$ immediately. Since $\p_2$ is already the last priceable vertex, we obtain $\mathcal{R}_1(\text{\ref{desc:option3}}) = 13$.
	To compare $\mathcal{R}_1(\text{\ref{desc:option1}})$ and $\mathcal{R}_1(\text{\ref{desc:option3}})$, we must evaluate $\mathcal{R}_2$. Since $\p_2$ is the last priceable vertex, the lower  and upper bound coincide and the only relevant option is $\text{\ref{desc:option1}}[2]$. Hence, we obtain $\mathcal{R}_2= \olb[i] = \oub[i] = 1$.
	Now, we can tell that $\text{\ref{desc:option3}}[1]$ is superior and hence the optimal choice for the leader is including $\p_1$ for its upper bound by excluding $\p_2$ (i.e., set $\price_2$ to the upper bound). Note that the actual optimal price for $\p_2$ is not the upper bound computed based on the benchmark prices but must rather be re-computed after plugging in the actual choice for $\price_1$.
	
	\begin{algorithm}[tbph]
	\algorithmInit
	\Fn{\ResolvePrices{$\pathGraph$, optimal pricing scheme}}{	
		\KwIn{A path $\pathGraph = (v_{1}, \dots, v_{n})$ with $k$ priceable vertices $p_1, \dots,p_k$, vertex weights $\omega(v)$ for $j \in F$ and optimal options $O[i]$ for each priceable vertex $\p_i$.}
		\KwOut{Optimal prices $\price_1,\dots,\price_k$ for the leader}
		\For{$i=1$ \To $k$}{
			$(\olb[i], \oub[i])$ = \ComputeBounds{$\pathGraph_{\prec \p_{i+1}}$} with true prices $\price_1, \dots, \price_{i-1}$\;
			\uIf{$O_2$ is the best option for vertex $\p_i$}{
				$\price_{i}$ = $\olb[i]$\;
			}
			\Else(\tcp*[f]{$O_1[i]$, $O_3[i]$ or $O_3[i-1]$ led to $\price_i$}){
				$\price_{i}$ = $\oub[i]$\;
			}
		}
		\KwRet{$\price_1, \dots, \price_k$}
	}
	\setstretch{1}
	\caption{An algorithm which computes the optimal prices given an optimal pricing scheme.}
	\label{alg:pricing}
\end{algorithm}
	\paragraph{Algorithm~\ref{alg:pricing} \texttt{ResolvePrices}.}
	The optimal pricing scheme corresponds to optimal options for in- or excluding certain vertices from the follower's \ac{MWVC}. We achieve this  by setting the prices to lower or upper bounds. The bounds computed using the benchmark suffice to compare the different options, yet, they do not yield the actual prices. To compute these, we must consider the actual prices of $\p_j$ for every $j < i$ when computing the bounds of $\p_i$. We do this in \Cref{alg:pricing}.
	
	\section{Proof of Theorem~\ref{thm:StackVCPlinTime}}
	\label{sec:ProofOfMainContr}
	In the following, we bring together all our insights from the previous section to prove our main contribution Theorem~\ref{thm:StackVCPlinTime}. The foundation of our algorithm are the bounds computed by the subroutine \ComputeBounds (cf.\  Algorithm~\ref{alg:ComputeBounds}). Definition~\ref{def:BoundIntuition} states the impact of these bounds on a \ac{MWVC}. The remaining statements in Section~\ref{sec:cutting+bounds}, in particular Lemma~\ref{lem:correctness_boundAlgo}, justify their computation by Algorithm~\ref{alg:ComputeBounds}. The concept of these bounds rests upon our insights from Section~\ref{sec:in_cover} in which we dissected the structure of \acp{MWVC} of a path. The fundamental concept here is the \ref{eq:cover_condition} introduced in Lemma~\ref{lem:in_cover}. The \ref{eq:cover_condition} allows us to determine when a vertex is part of a \ac{MWVC} and thereby exhaust the potential revenue of a priceable vertex. Section~\ref{sec:effect} is devoted to solving the primary challenge of \ac{StackVC} on a path, namely the combinatorial problem of finding the optimal selection of priceable vertices to be included in the follower's solution. Lemma~\ref{lem:beyond_bounds} and \ref{lem:effectRange} justify our restriction of the leader's choices to the previously introduced bounds. We formalized this in Lemma~\ref{lem:options_OnePart} and \ref{lem:Option_twoPartition}. 
	Lemma~\ref{lem:effectRange} also provides important insights on the interaction of priceable vertices which we use to appraise the revenue of the different options for each priceable vertex (cf.\ Lemma~\ref{lem:RevenueEstimation_onePart} and \ref{lem:RevenueEstimation_oppPart}). An important step in this procedure was the definition of a benchmark price to make the options artificially comparable. Section~\ref{subsec:comparing_options}, in particular, Algorithm~\ref{alg:comparing}, provides the details of how this comparison can be implemented efficiently. Once the optimal pricing decisions have been determined based on the benchmark prices, we must compute the actual corresponding prices. We do this with Algorithm~\ref{alg:pricing}. 
	
	Algorithm~\ref{alg:StackVCP_sketch} aggregates all the above steps and together with the aforementioned statements attesting its correctness, it serves as proof of Theorem~\ref{thm:StackVCPlinTime} with one small caveat that we discuss in the following.
	
	\paragraph{Time complexity.}
	It is easy to see that Algorithm~\ref{alg:StackVCP_sketch} in its present implementation has a time complexity in $\mathcal{O}(n^2)$. Responsible for the quadratic runtime, however, is only the computation of the bounds in the subroutine \ComputeBounds (cf.\ Algorithm~\ref{alg:ComputeBounds}) which runs in $\mathcal{O}(n \cdot k)$. A savvier implementation of this subroutine gets along with less than two traversals of the path, thereby achieving (strongly) linear complexity. The enhancement bases on the following observation:
	
	When setting the $\price_i$ to $\olb[i](\price_1,\dots, \price_{i-1})$ or to $\oub[i](\price_1,\dots, \price_{i-1})$, we reach a vertex $v$ at which $\gap(v) = 0$. We can cut the path after $v$, and compute the remaining bounds on $\pathGraph_{v \prec}$. The same holds if we set $\price_i$ to $\olb[i]$ or $\oub[i]$, when executing \ResolvePrices. In the worst case, $v$ is always close to $\p_i$ and the bounds never intersect, so we always reach $\p_{i+1}$ before we know that we must cut at $v$ and therefore traverse $v, \dots, p_{i+1}$ two times.
	
	With this optimization, we achieve the claimed strongly linear runtime since the branching scheme clearly runs in $\mathcal{O}(n)$ as well. Since we only store a constant number of variables for every priceable vertex, storage complexity is (strongly) linear as well.
	
	\section{Conclusion and Outlook}
	
	The main motivation of this paper is answering whether \ac{StackVC} can be solved on some subclass of bipartite graphs, like paths or trees -- a question that was also raised in \cite{StackMaxClosureOMS}. Indeed, we can affirmatively answer this question for paths, as we state in Theorem~\ref{thm:StackVCPlinTime}. Thereby our algorithm is the first step to closing the complexity gap of \ac{StackVC} on bipartite graphs. It was already known to be \textsc{NP}-hard on general bipartite graphs~\cite{StackMaxClosureOMS} and polytime solvable with the restriction that priceable vertices may only be on one side of the bipartition~\cite{BriestEtAl_NetworkPricing}. However, to the best of our knowledge, there have not been any positive results regarding the complexity of \ac{StackVC} on bipartite graphs without any restrictions.
	The question whether \ac{StackVC}  can be solved efficiently also on trees remains open. 
	
	\bibliographystyle{splncs04}
	\bibliography{ms.bib}
\end{document}